\documentclass[%
 reprint,
superscriptaddress,
 amsmath,amssymb,
 aps,
]{revtex4-2}
\usepackage{graphicx}
\usepackage{indentfirst}
\usepackage{braket}
\usepackage{float}
\usepackage{amsmath}
\usepackage{amssymb}
\usepackage{amsfonts}
\usepackage{tikz-cd}
\tikzcdset{
  row sep=1.1em,
  column sep=1.1em,
  cramped,
  arrows=semithick
}
\usepackage{graphicx}     
\usepackage{subcaption}
\usepackage{hyperref}
\usepackage{natbib}
\usepackage{float}
\usepackage{amsthm}
\usepackage{epstopdf}
\newtheorem{theorem}{Theorem}
\newtheorem{corollary}{Corollary}[theorem]
\usepackage{ragged2e}
\usepackage{CJK}
\usepackage{esint}
\usepackage{color}
\usepackage[T1]{fontenc}
\usepackage{amsfonts}
\usepackage{natbib}
\usepackage{footmisc}
\usepackage{verbatim}
\usepackage{multirow}
\usepackage[english]{babel}
\usepackage{url}
\usepackage{bm}
\usepackage{hyperref}
\definecolor{darkblue}{rgb}{0,0,0.5}
\hypersetup{
colorlinks=true,
linkcolor=black,
filecolor=blue,
citecolor=darkblue,  
urlcolor=black,
}
\newtheorem{definition}{Definition}

\usepackage{tikz}
\usetikzlibrary{arrows.meta,positioning,calc,fit,decorations.pathmorphing,shapes.geometric}

\usepackage{enumitem,kantlipsum}
\usepackage{refcount}

\newcounter{propositioncounter}
\newenvironment{proposition}[1][]{\refstepcounter{propositioncounter}
   {{\em Proposition~\thepropositioncounter #1}.---} \rmfamily}
   
\def\be{\begin{equation}}
\def\ee{\end{equation}}
\def\ba{\begin{eqnarray}}
\def\ea{\end{eqnarray}}
\newcommand{\tr}{{\rm Tr}}

\newcommand{\eq}[1]{(\hyperref[eq:#1]{\ref*{eq:#1}})}

\renewcommand{\sec}[1]{\hyperref[sec:#1]{Section~\ref*{sec:#1}}}
\newcommand{\thrm}[1]{\hyperref[thm:#1]{Theorem~\ref*{thm:#1}}}
\newcommand{\lemm}[1]{\hyperref[lemm:#1]{Lemma~\ref*{lemm:#1}}}
\newcommand{\pro}[1]{\hyperref[pro:#1]{Proposition~\ref*{pro:#1}}}
\newcommand{\corr}[1]{\hyperref[corr:#1]{Corollary~\ref*{corr:#1}}}
\newcommand{\deff}[1]{\hyperref[deff:#1]{Definition~\ref*{deff:#1}}}
\newcommand{\fig}[1]{\hyperref[fig:#1]{\ref*{fig:#1}}}
\newcommand{\tbl}[1]{\hyperref[fig:#1]{\ref*{tbl:#1}}}

\usepackage{graphicx}
\usepackage{dcolumn}
\usepackage{bm}

\begin{document}

\preprint{APS/123-QED}

\title{Even‐Odd Splitting of the Gaussian Quantum Fisher Information: From Symplectic Geometry to Metrology}

\author{Kaustav Chatterjee}

 \email{kauch@dtu.dk}
\affiliation{
Center for Macroscopic Quantum States (bigQ), Department of Physics, Technical University of Denmark, \\
 Building 307, Fysikvej, 2800 Kongens Lyngby, Denmark
}%
\author{Tanmoy Pandit}
\affiliation{QMill Oy Keilaranta 12 D, 0215 0 Espoo, Finland}
\affiliation{Institute for Theoretical Physics, Leibniz Institute of Hannover, Hannover, Germany}
\affiliation{Institute for Physics and Astronomy, TU Berlin, Germany}
\author{Varinder Singh}
\affiliation{School of Physics, Korea Institute for Advanced Study, Seoul 02455, Korea}
\author{Pritam Chattopadhyay}

\affiliation{Department of Chemical and Biological Physics,
Weizmann Institute of Science, Rehovot 7610001, Israel}
\author{Ulrik Lund Andersen}%
\email{ulrik.andersen@fysik.dtu.dk}
\affiliation{
Center for Macroscopic Quantum States (bigQ), Department of Physics, Technical University of Denmark, \\
 Building 307, Fysikvej, 2800 Kongens Lyngby, Denmark
}

\date{\today}

\begin{abstract}
We introduce a canonical decomposition of the quantum Fisher information (QFI) for centered multimode Gaussian states into two additive pieces: an \textit{even part} that captures changes in the symplectic spectrum and an \textit{odd part} associated with correlation-generating dynamics. On the pure-state manifold, the even contribution vanishes identically, while the odd contribution coincides with the QFI derived from the natural metric on the Siegel upper half-space, revealing a direct geometric underpinning of pure-Gaussian metrology. This also provides a link between the graphical representation of pure Gaussian states and an explicit expression for the QFI in terms of graphical parameters. For evolutions completely generated by passive Gaussian unitaries (orthogonal symplectics), the odd QFI vanishes, while thermometric parameters contribute purely to the even sector with a simple spectral form; we also derive a state-dependent lower bound on the even QFI in terms of the purity-change rate. We extend the construction to the full QFI matrix, obtaining an additive even–odd sector decomposition that clarifies when cross-parameter information vanishes. Applications to unitary sensing (beam splitter versus two-mode squeezing) and to Gaussian channels (loss and phase-insensitive amplification), including joint phase–loss estimation, demonstrate how the decomposition cleanly separates resources associated with spectrum versus correlations. The framework supplies practical design rules for continuous-variable sensors and provides a geometric lens for benchmarking probes and channels in Gaussian quantum metrology.
\end{abstract}

 \maketitle


\section{\label{sec:level1}Introduction\protect}

Precise measurement and parameter estimation underpin some of the most fundamental tasks in physics, from probing weak forces to stabilizing clocks and characterizing quantum devices. The framework of quantum metrology provides statistical limits on such tasks by combining quantum mechanics with estimation theory, and has grown into a central subfield of quantum information science \cite{Giovannetti_2011,Helstrom1969QuantumDA,Holevo1982ProbabilisticAS,paris2009quantumestimationquantumtechnology,Gio1,Liu,Toth2014QMetrology,Pezze2018AtomicEnsembles,Deng2024QuantumEnhancedFock}. Key to the analysis of sensing or estimation protocols using quantum mechanics is the \textit{quantum Fisher information} (QFI)~\cite{Liu,meyer2021fisher,rath2021quantum,frowis2016detecting,vitale2024robust}, which quantifies the sensitivity of a quantum state to changes in an encoded parameter and determines, via the quantum Cramér–Rao bound~\cite{paris2009quantumestimationquantumtechnology,Helstrom1969QuantumDA}, the best achievable scaling of estimation error for any measurement and estimator.

Continuous-variable (CV) systems and, in particular, Gaussian states, form a natural arena for testing and realizing ideas of quantum metrology~\cite{Oh2020OpticalEstimation,PhysRevA.88.040102,Rahman2025GenuineNonGaussianity,Santos2025EnhancedFrequency,santos2025improving,hall2025probing}. Bosonic modes of light, microwaves, phonons, and collective excitations in solid-state platforms are routinely prepared and controlled in regimes where their states are well approximated by Gaussians \cite{Serafini2017QuantumCV,weed,Adesso_2014}. Many flagship applications of quantum sensing can be modeled in this language: interferometric phase estimation and gravitational-wave detection \cite{inf1,inf2,PhysRevX.13.041021}, continuous-variable quantum key distribution and device certification \cite{zhang2024continuous,usenko2025continuous}, thermometry and noise spectroscopy \cite{cenni,mauro,cenni2022thermometry,Chattopadhyay_2025,gsh7r7ms,seah2019collisional,mirkhalaf2024operational,Paz_Silva_2017}, and even sensing tasks in biology and chemistry \cite{deAndrade:20,casacio2021quantum,Aslam2023QuantumSensorsBio}. In all these settings, Gaussian states and Gaussian transformations provide a flexible testbed that is analytically tractable yet rich enough to capture nontrivial resources such as squeezing, multimode correlations, and thermal noise.

For Gaussian states, several powerful expressions for the QFI and for the closely related Bures metric are known \cite{Banchi_2015,monras2013phasespaceformalismquantum, quantumparameterestimationusing,Jiang2014ExponentialForm,Safranek2019EstimationGaussian,Marian2016QFIMTwoMode}. These formulas exploit the compact description of Gaussian states in terms of first and second moments and provide closed forms in terms of the covariance matrix and displacement vector. They are widely used in the literature and have enabled numerous case studies in Gaussian quantum metrology. However, in their standard form, these expressions tend to mix the contribution of physically distinct resources: changes in populations (or symplectic eigenvalues) versus changes in correlations between modes, and active (squeezing-type) versus passive (linear-optical) transformations. From the perspective of designing or benchmarking metrological protocols, it is precisely this distinction that is most informative: how much of the achievable sensitivity originates from \textit{thermal/spectral} aspects of the probe, and how much arises from genuinely correlation-generating dynamics such as squeezing or entanglement. In their standard form, however, existing Gaussian QFI expressions interwine these contributions, obscuring the physical origin of metrological advantage and limiting their usefulness as design principles.

At the same time, the manifold of Gaussian states carries a rich geometric structure. Centered Gaussian states are fully characterized by their covariance matrix, which is a real, symmetric, positive matrix constrained by the uncertainty principle~\cite{Adesso_2014}. The natural symmetry group acting on these covariance matrices is the real symplectic group $Sp(2n,R)$, corresponding to Gaussian unitary transformations on $n$ modes \cite{Arvind_1995}. For pure Gaussian states, this geometry simplifies further: they can be represented by a complex adjacency (or \textit{graph}) matrix $Z$ in the Siegel upper half-space~\cite{Ohsawa2015Siegel}, and Gaussian unitaries act by \textit{Möbius transformations} on $Z$ \cite{Menicucci_2011}. This graphical calculus has proved extremely useful in the theory of Gaussian cluster states and measurement-based quantum computation \cite{weed}. More generally, mixed Gaussian states can be described by their symplectic eigenvalues and the symplectic transformation that diagonalizes the covariance matrix. This naturally leads to an orbit structure and, as we formalize in this work, to a fiber-bundle~\cite{7a97cef4-8443-3a57-aced-2037f84b9e06} viewpoint in which the \textit{spectral data} and the \textit{symplectic frame} play distinct roles. Despite this rich geometric background, the connection between such structures and the QFI for Gaussian states has not been fully exploited.

In this work, we bridge these perspectives by showing that the QFI (equivalently, the Bures metric) for centered multimode Gaussian states admits a canonical orthogonal decomposition into two additive pieces, which we call the even and odd contributions. The split is defined at the level of tangent vectors to the Gaussian state manifold, using the Cartan decomposition of the symplectic Lie algebra and the parity of matrices with respect to the symplectic form~\cite{hall2000elementaryintroductiongroupsrepresentations}. This yields a decomposition of any infinitesimal variation of the covariance matrix into an \textit{even component}, which is naturally associated with changes in the symplectic spectrum, and an \textit{odd component}, which is associated with correlation-generating dynamics. We prove that the Bures metric is block-diagonal with respect to this splitting, so that the total QFI is a sum of two nonnegative, geometrically defined terms, with no cross term.

On the pure-state manifold, all symplectic eigenvalues are fixed to $1/2$, so variations in the covariance matrix are purely odd in this sense. We show that in this case the even contribution vanishes identically and the odd contribution coincides with the QFI derived from the natural Riemannian metric on the Siegel upper half-space~\cite{Braunstein1994StatisticalDistance,Petz1996Geometries,Bengtsson2017Geometry}. Equivalently, we obtain an explicit graph-based expression for the Bures metric of pure Gaussian states in terms of the complex adjacency matrix $Z$, thereby providing a direct bridge between graphical calculus and Gaussian metrology. This establishes a clear geometric underpinning of pure-Gaussian QFI and shows that, for pure states, the entire metrological content is captured by the geometry of the Siegel domain.

For general (mixed) centered Gaussian states, we introduce a bundle structure over the space of covariance matrices that separates spectral data and symplectic frames. Within this framework, we characterize the even–odd splitting intrinsically, relate it to the Cartan decomposition of the symplectic algebra~\cite{Helgason2001DifferentialGeometry,Knapp2002LieGroupsBeyond,Salamon2024Notes,ReichsteinCartanIwasawa}, and prove that the resulting even and odd contributions to the Bures metric are invariant under a large class of symplectic frame changes. We then derive simple closed-form formulas for both even and odd contributions in a \textit{Williamson eigenbasis}, expressed directly in terms of the symplectic eigenvalues and the blocks of the infinitesimal generator.

These geometric constructions have several concrete operational consequences. First, we show that for any parameter that only changes the symplectic eigenvalues (what we call a thermometric parameter), the odd sector automatically vanishes and the full QFI is exhausted by the even part. In this case, even QFI takes a particularly simple spectral form, depending only on derivatives of the symplectic eigenvalues. Second, we establish a state-dependent lower bound on even QFI in terms of the rate of change of the state’s purity, showing that such a contribution to QFI controls how fast the state is driven away from or towards purity. In particular, this bound diverges when approaching the pure-state manifold along directions that change the symplectic eigenvalues, which matches the behaviour observed in transmissivity estimation~\cite{Safranek:2016blj}. Taken together, these results motivate interpreting the even sector as quantifying \textit{thermodynamic or noise-related} aspects of metrology and as a witness of purity-breaking along the estimation path. Third, we prove that passive Gaussian unitaries—those represented by orthogonal symplectic matrices, corresponding to \textit{linear interferometers} without squeezing—do not generate any odd contribution: their associated odd QFI vanishes. Conversely, genuinely active Gaussian operations, such as two-mode squeezing, can contribute nontrivially to the odd QFI by creating or redistributing correlations at fixed symplectic spectrum. This gives the odd sector a clear operational meaning as a quantifier of correlation-generating resources and provides a clean way to distinguish sensing protocols that rely purely on population changes from those that exploit active Gaussian dynamics.

We also extend our construction to the full QFI matrix for multi-parameter estimation, obtaining an additive decomposition into even and odd matrix contributions. In this setting, we show a sector decoupling: parameters whose generators lie purely in the even sector do not mix, at the level of QFI, with parameters whose generators lie purely in the odd sector. This clarifies when \textit{cross-parameter information} (and thus potential incompatibilities) can or cannot arise between thermometric parameters and correlation-generating parameters. We illustrate this in joint estimation scenarios, such as simultaneous sensing of phase and loss, where our framework cleanly separates the roles of population changes versus correlations in the attainable precision and in the structure of the QFI matrix.

Finally, we demonstrate the usefulness of the even–odd splitting on several representative examples. We first analyze unitary sensing tasks involving beam splitters and two-mode squeezers, showing how the graphical calculus gives immediate access to QFI on the pure-state manifold and how the odd contribution discriminates between passive and active encodings. We then study Gaussian channels such as loss and phase-insensitive amplification, both individually and in joint phase–loss estimation schemes. In these examples, the even sector captures changes in the noise level and purity, while the odd sector captures the build-up of correlations and interference effects. 

The rest of the paper is organized as follows. In Sec. II we review the basics of multimode Gaussian states, the symplectic group and the graphical representation of pure Gaussian states, and we introduce the fiber-bundle structure over the Gaussian state space. In Sec. III we define the even–odd splitting of tangent vectors, derive the corresponding decomposition of the Bures metric and QFI, and obtain closed-form expressions in the Williamson frame, together with invariance properties. In Sec. IV we apply our framework to concrete examples of unitary sensing and Gaussian channels, including thermometry and joint phase–loss estimation, highlighting the operational meaning of the even and odd sectors. We conclude in Sec. V with a discussion of how this geometric perspective can be used to guide the design and benchmarking of continuous-variable sensing protocols and outline possible extensions beyond Gaussian states.

\section{\label{sec:level2}Preliminaries\protect}
An $N$-mode bosonic continuous-variable system is described by annihilation operators $\left\{\hat{a}_k, 1\le k \le N\right\}$, which satisfy the commutation relation $\left[\hat{a}_k,\hat{a}_j^\dagger\right]=\delta_{kj}, \left[\hat{a}_k,\hat{a}_j\right]=0$. Equivalently, one can define 2N real quadrature field operators $\hat{q}_k=\frac{1}{\sqrt{2}}(\hat{a}_k+\hat{a}_k^\dagger), \hat{p}_k=\frac{i}{\sqrt{2}}\left(\hat{a}_k^\dagger-\hat{a}_k\right)$ and collect them into the real vector
$\hat{x}=\left(\hat{q}_1,\hat{p}_1,\cdots, \hat{q}_N,\hat{p}_N\right)^T$. This vector satisfies the canonical commutation relation
$
\left[\hat{x}_i,\hat{x}_j\right]=i {  \Omega}_{ij}.
$
where ${  \Omega}=i \bigoplus_{k=1}^N { \sigma_y} $ and  ${ \sigma_y}$ is the Pauli matrix. We will also sometimes change the ordering to $(\hat q_1,\hat q_2\ldots,\hat p_1,\hat p_2,\ldots)$ which leads to permutation of the $\Omega$ as $\begin{pmatrix}
    0& I\\
    -I & 0\\
\end{pmatrix}$. A quantum state $\rho$ can be conveniently described by its (symmetrically ordered) characteristic function
\be 
\chi\left({  \xi};\rho\right)=\tr \left[\rho \hat{D}\left({  \xi}\right)\right],
\label{Wigner_characterisitic_function}
\ee
where 
$
\hat{D}\left({  \xi}\right)=\exp\left(i{  \xi}^T{  \Omega} \hat{x}  \right)$
is the multi-mode Weyl displacement operator and $  \xi=\left(\xi_1,\cdots \xi_{2N}\right)^T\in \mathbb{R}^{2N}$ is a phase-space vector. A state $\rho$ is Gaussian if and only if its characteristic function has the Gaussian form~\cite{weed,Adesso_2014,Holevo1982ProbabilisticAS,brask2021gaussian}
\be
\chi \left({  \xi}_E;\rho\right)=\exp\left(-\frac{1}{4}{  \xi}^T \left({  \Omega^T} {  \Lambda } {  \Omega}\right){  \xi}+i {  \xi}^T{  \Omega^T} \overline{  x}\right).
\ee
Here $\overline{  x}=\braket{  \hat{x}}_{\rho}$ is the state's mean and 
$
{  \Lambda}_{ij}=\braket{\{\hat{x}_i-\overline{x}_i,\hat{x}_j-\overline{x}_j\}}_\rho$
is its covariance matrix, with $\{,\}$ denoting the anticommutator. 
Thus, every Gaussian state is completely characterized by $\overline{  x}$ and $\Lambda$. For the rest of the paper, we will only consider centered Gaussians (meaning the first moment is 0), leaving the idea of defining appropriate splitting for \textit{uncentered Gaussians open}. Hence, our set of Gaussian states is in one-to-one correspondence with the set of covariance matrices, which we define as a set $\mathcal{G}_n$
\begin{equation}
    \mathcal{G}_n:= \{V\in \mathbb{R}^{2n\times 2n}| V=V^T ; V+i\Omega/2\geq 0\},
\end{equation}
where we work in ordering $(q_1,..,q_n,p_1,...,p_n)$ and the symplectic form is given by $\Omega:=\begin{pmatrix} 0 & I_n\\
-I_n&0\end{pmatrix}$. By Williamson's decomposition, any covariance matrix can be written as $V=SDS^T$ where $D$ contains the symplectic eigenvalues and $S$ belongs to $Sp$, the symplectic group $Sp(2n,R)$. Explicitly, $D=diag(K,K)$ with $K=diag(k_1,k_2,\ldots, k_n)$. But the multiplicity of these symplectic values would play a role later, and hence, we define $\mathcal G^\mu_n$ as those $V=SDS^T$ for which $D$ has a multiplicity pattern of $\mu$ denoted by $d(D)=\mu$. For example, for a $3$ mode state with multiplicity pattern $\mu=(1,2)$ means $K=(k_1,k_2,k_2)$ and $k_1\leq k_2$. This gives $\mathcal{G}_n=\sqcup_{\mu}\mathcal{G}^\mu_n$. We think of this as a stratified manifold. We can also define the set $\mathcal C_n^\mu:=\{D|d(D)=\mu\}$. This set is the collection of all Williamson diagonalized covariance matrices which have a degeneracy pattern of $\mu$. Given such a $D$, we have the symplectic stabilizer of it given by $Stab(D):=\{S\in Sp|SDS^T=D\}$. This stabilizer set is equivalent to the group manifold $\mathcal F_n^\mu:=\Pi_{\mu_i}U(\mu_i)$, which intuitively is a bunch of phase shifters that work within the degenerate subspaces. The notation $U(d)$ implies unitary matrices of dimension $d^2$. The equivalence as stated comes from the fact of the equivalence of orthogonal symplectic matrices and unitary matrices as stated in \cite{Arvind_1995}. We can also define the orbit of $D$ as $\mathcal O_D:=\{SDS^T|S\in Sp\}=Sp/\mathcal{F}^\mu_n$ where the last equality is understood as a collection of equivalence classes where you upto the phase shifters that can act on the degenerate blocks. 
\begin{table*}[t]
  \centering
  \caption{Summary of spaces and notation used in the geometric description of Gaussian states.}
  \label{tab:spaces}
  \begin{ruledtabular}
  \begin{tabular}{lll}
    Symbol & Type & Description \\ \hline
    $\mathcal{G}_n$ &
      Subset of $\mathrm{Sym}(2n,\mathbb R)$ &
      Covariance matrices of centered $n$-mode Gaussian states. \\[0.25em]

    $\mathcal{G}_n^\mu$ &
      Stratum of $\mathcal G_n$ &
      Covariances with fixed symplectic-eigenvalue multiplicity pattern $\mu$. \\[0.25em]

    $\mathcal C_n^\mu$ &
      Subset of $\mathcal G_n^\mu$ &
      Williamson-diagonal covariances $D = \mathrm{diag}(K,K)$ with
      multiplicity pattern $\mu$. \\[0.25em]

    $\mathfrak h_n$ &
      Siegel upper half-space &
      Complex symmetric matrices $Z = X + iY$ with $X=X^T$, $Y=Y^T>0$. \\[0.25em]

    $Sp$ &
      Lie group &
      Real symplectic group : $S^T \Omega S = \Omega$ (Gaussian unitaries). \\[0.25em]

    $O_{sp}$ &
      Lie subgroup of $Sp$ &
      Orthogonal symplectic group; passive Gaussian unitaries; $O_{sp} \cong U(n)$. \\[0.25em]

    $U(n)$ &
      Lie group &
      Standard unitary group; identified with $O_{sp}$ via a fixed isomorphism. \\[0.25em]

    $\mathcal F_n^\mu$ &
      Lie subgroup of $U(n)$ &
      Stabilizer of $D\in C_n^\mu$:
      $\mathrm{Stab}(D) \cong \mathcal F_n^\mu := \prod_i U(\mu_i)$. \\[0.25em]

    $\mathcal O_D$ &
      Homogeneous space &
      Symplectic orbit of $D$:
      $O_D = \{ S D S^T \mid S\in\mathrm{Sp}\}
       \cong Sp/\mathcal F_n^\mu$. \\[0.25em]

    $Z$ &
      Element of $\mathfrak h_n$ &
      Complex adjacency (graph) matrix of a pure Gaussian state. \\[0.25em]

    $\Gamma$ &
      Spectral/degeneracy data &
      Element of $\mathcal C_n^\mu \times U(n)/\mathcal F_n^\mu$ collecting
      symplectic eigenvalues and residual passive mixing. \\[0.25em]

    $(Z,\Gamma)$ &
      Quantum state label &
      Mixed-state parametrization,
      pure states correspond to $(Z,I/2)$. \\
  \end{tabular}
  \end{ruledtabular}
\end{table*}

\subsection{Fiber bundle structure over Gaussian state space}
As discussed in the introduction, we want to build up to a formalism of Gaussian states which includes graphical information (explicitly defined later) as well as remaining information (not captured by graphs). To rigorously capture these pieces of information into a single structure, we resort to a fiber-bundle approach, which explicitly shows how much more freedom is present in the Gaussian mixed states compared to pure states. We briefly explain to the readers what a fiber bundle is and point them towards references like~\cite{Lang1999,Steenrod1951FibreBundles,Husemoller1994FibreBundles,KobayashiNomizu1963FoundationsI,Nakahara2003GeometryTopologyPhysics,NashSen1983Topology,Frankel2011GeometryOfPhysics,Bleecker2005GaugeTheory}. 

\begin{definition}[Fiber bundle]
Let $E$ and $B$ be smooth manifolds and let $F$ be another manifold (the \emph{typical fiber}).
A \emph{(smooth) fiber bundle} with total space $E$, base space $B$ and fiber $F$ is a
surjective smooth map
\begin{equation}
  \pi : E \to B,
\end{equation}
such that locally around any point of $B$ the total space looks like a product $U\times F$,
and the map $\pi$ is locally just projection onto $U$.
\end{definition}
A fiber bundle is said to be \emph{globally trivial} (or simply a \emph{trivial bundle}) if there exists a global bijective (smooth) map
\begin{equation}
  \Phi : E \xrightarrow{\ \cong\ } B \times F.
\end{equation}
In such cases, we often identify $E$ with $B\times F$ and write $E \cong B\times F$.

\begin{definition}[Section]
Given a fiber bundle $\pi : E \to B$, a \emph{(smooth) section} is a smooth map
\begin{equation}
  \sigma: B \to E,
\end{equation}
such that
\begin{equation}
  \pi \circ \sigma= I.
\end{equation}
\end{definition}

In this paper, we will only encounter fiber bundles that are globally trivial at the level of
the strata we consider. Throughout the paper, we use the notation $\pi$ for bundle projections, $s$ or $\sigma$ for sections, and the symbol $\cong$ to denote diffeomorphisms (smooth and bijective maps) of manifolds.

The way we have defined the various spaces makes our setting to admit a natural fiber bundle structure with total space: $\mathcal G^\mu_n$, base space: $\mathcal C^\mu_n$ and the fiber as $Sp/\mathcal F_n^\mu$ with the continuous surjection map $\pi_g:\mathcal G^\mu_n\to\mathcal C^\mu_n$ defined as $\pi_g(SDS^T)=D$ and a section $\sigma_g: C^\mu_n\to\mathcal G^\mu_n $ as $\sigma_g(D)=D$. Now the idea would be to see this fiber as a trivial one by inducing a different section, which comes from looking at $Sp$. For doing so, the graphical calculus of Gaussian pure states comes in handy, which we review here very abstractly. The main idea is to consider a fiber bundle with total space: $Sp$, base space: $\mathfrak h_n$ and the fiber as $O_{sp}$. Here, $\mathfrak h_n:=\{ X+iY| Y>0; X=X^T; X\in\mathbb{R}^{n\times n}; Y\in\mathbb{R}^{n\times n}  \}$ is called the Siegel upper half-space and $O_{sp}$ is $Sp \cap O(2n,R)$ where $O(2n,R)$ is the orthogonal group. $O_{sp}$ basically contains the passive optics elements and is isomorphic to $U(n)$\cite{Arvind_1995}. This fiber bundle structure has an interesting choice of section $\sigma_s: \mathfrak h_n\to Sp$ which comes from the pre-Iwasawa decomposition of Symplectic matrices. To make this explicit, we state the decomposition as 
\be
S=\begin{pmatrix} I & 0\\
X& I\end{pmatrix}\begin{pmatrix} Y^{-1/2}& 0\\
0& Y^{1/2}\end{pmatrix}O=S_Z O,
\ee
where $S\in Sp$, $X+iY\in \mathfrak h_n$ and $O\in O_{sp}$. This naturally gives us the desired section $\sigma_s(Z)=S_Z$ as well as the continuous surjection map $\pi_s: Sp\to\mathfrak h_n$ as $\pi_s(S_ZO)=Z$. This section induces a global trivialization of the fiber bundle as $Sp\cong \mathfrak h_n\times O_{sp}$. Here, trivialization means that our total space breaks into a direct product of the base space and the fiber. The other fiber is also similarly a trivial one by the fact that $\mathcal{G}^\mu_n\cong \mathcal C^\mu_n\times Sp/\mathcal F^\mu_n$. Now, utilizing the trivialization of $Sp$ we conclude that 
\be
\mathcal{G}^\mu_n\cong \mathcal C^\mu_n \times \mathfrak h_n \times U(n)/\mathcal F^\mu_n,
\ee
This isomorphism serves as an extension of the Siegel upper half-space to include even the mixed states. More precisely, we associate a tuple of information to our state $(Z,\Gamma)$ where $Z\in\mathfrak h_n$ and $\Gamma\in \mathcal{C}^\mu_n\times U(n)/\mathcal{F}^\mu_n$. It is important to note that, given a fixed state, we also fix the multiplicity pattern $\mu$, but the entire set of Gaussian states is still a disjoint union over $\mathcal G^\mu_n$. This representation of Gaussian states has the nice property that $\{(Z,\frac{I}{2})|Z\in \mathfrak h_n\}$ describes the entire set of pure states, and viewing $Z$ as a complex adjacency matrix for a state recovers the standard graphical calculus \cite{Menicucci_2011}. Further, we also state the general transformation rule of such a representation under a Gaussian unitary.
\begin{figure}[h]
\begin{centering}
\includegraphics[scale=0.5]{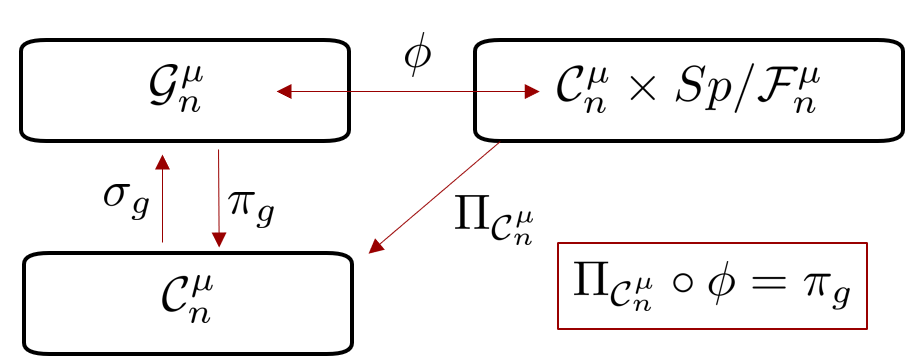}
\caption{\label{fig:epsart} Various spaces and their commutative diagram. Here, $\phi$ is the map that defines the trivialization. This is for the set of Gaussian states.}
\end{centering}
\end{figure}
\begin{figure}[h]
\begin{centering}
\includegraphics[scale=0.5]{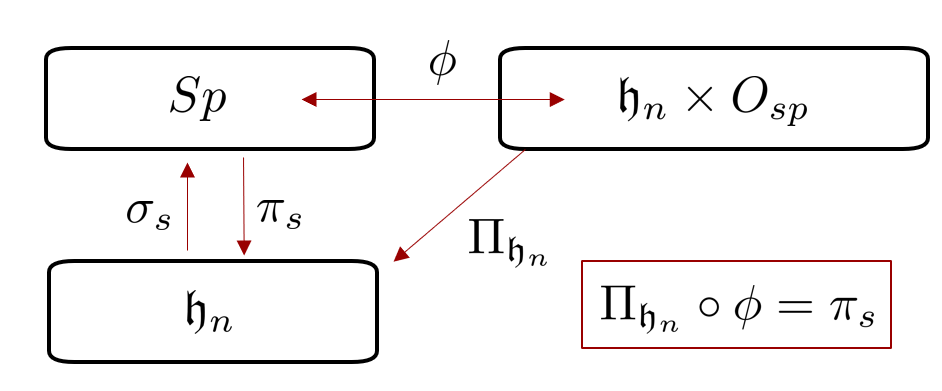}
\caption{\label{fig:epsart} Various spaces and their commutative diagram. Here, $\phi$ is the map that defines the trivialization. This is for the Symplectic group.}
\end{centering}
\end{figure}
\begin{theorem}
    Consider the symplectic matrix $T=\begin{pmatrix} A & B\\
C& D\end{pmatrix}$ corresponding to a Gaussian transformation, then the state transformation $(Z,\Gamma)\to_T (Z',\Gamma')$ is given by:
\begin{itemize}
    \item $Z'=(C+DZ)(A+BZ)^{-1}$,
    \item $\Gamma'=O\Gamma O^T$ with $O=S^{-1}_{Z'}TS_Z\in O_{sp}$.
\end{itemize}
\end{theorem}
\begin{proof}
    Given a covariance matrix $V=S_Z\Gamma S_Z^T$ under Gaussian unitaries, it transforms as $V'=TS_Z\Gamma(TS_Z)^T$. Now $TS_Z$ is a new symplectic matrix and can again be decomposed (pre-Iwasawa) into $TS_Z=S_{Z'}O$. Now the action of the group $Sp$ over $\mathfrak h_n$ is given by the generalized Möbius transformation $Z' =(C+DZ)(A+BZ)^{-1}$\cite{Menicucci_2011}. Now realizing $\Gamma'=O\Gamma O^T$ completes the proof.  
\end{proof}

\section{\label{sec:level4}Splitting of Bures metric and QFI\protect}
For the Bures metric for Gaussian states, we take the definition as~\cite{Banchi_2015}
\begin{equation}
    ds^2=2(1-\mathcal{F}(\rho,\rho+d\rho))=\frac{1}{2}\tr[dV(4\mathcal{L}_V+\mathcal L_\Omega)^{-1}(dV)],
\end{equation}
where $\rho$ has covariance matrix $V$ and $\rho+d\rho$ has covariance matrix $V+dV$. Here, $\mathcal L_M(N):=MNM$ and the inverse is the pseudo-inverse. For convenience, we label the operator inside pseudo-inverse as $\mathcal{M}_V=4\mathcal L_V+\mathcal L_\Omega$. Also, $\mathcal F$ stands for the standard \textit{Uhlmann fidelity}. For pure states, we have $ds^2_{pure}=\frac{1}{8}\tr[(V^{-1}dV)^2]$ \cite{Banchi_2015}. Given that pure states are fully described by $(Z, I/2)$, the Bures metric must be completely described by graph parameters. Our first theorem asserts this.\\
\begin{theorem}\label{thanos1}
    For pure Gaussian states Bures metric is proportional to the standard Riemannian metric on the Siegel upper half-space given by $ds^2_{Siegel}=\tr[Y^{-1}dZY^{-1}dZ^*]$.
\end{theorem}

The proof is deferred to the appendix \ref{app1}. This is intuitively valid because the manifold of pure states is isomorphic to $\mathfrak h_n$. Bures metric also connects to the expression of QFI and we obtain it completely using graphical parameters:
\be
QFI(t)=\frac{1}{2}\tr[Y^{-1}\frac{d(Z)}{dt}Y^{-1}\frac{dZ^*}{dt}]
\ee
Physically, this equation bears a lot of importance as it conveys that if we initially start with a Gaussian graph state $Z$ and pass it through a network of Gaussian unitaries characterized by multiple parameters $\vec \theta=(\theta_i)_i$ then the output graph $Z(\vec\theta)$ completely quantifies the QFI. More precisely, we have the QFI matrix (QFIM) as:
\begin{equation}
    F_{ab}(\vec \theta)=\frac{1}{2}\tr[Y^{-1}(\vec\theta)\frac{d(Z(\vec\theta))}{d\theta_a}Y^{-1}\frac{d(Z^*(\vec\theta))}{d\theta_b}].
\end{equation}

For the next subsection, it is important to observe that $ds^2$ is invariant under Gaussian unitary transformations such that $V\to SVS^T$ where $S$ is some symplectic matrix. More precisely, we have,
\begin{equation}
\begin{split}
    ds^2&=2(1-\mathcal{F}(U\rho U^\dagger,U(\rho+d\rho)U^\dagger)),\\
    &=\frac{1}{2}\tr[SdV S^T(\mathcal M_{SVS^T})^{-1}(SdV S^T)].
    \end{split}
\end{equation}

\subsection{Splitting of Bures metric for $(Z,\Gamma)$ representation}
We would be dealing with variations over our state space parametrized by $(Z,\Gamma)$. Variations along the graph part would connect with variations over the symplectic group which lies within the Lie-algebra $\mathfrak{sp}_{2n}:=\{X|X\Omega+\Omega X^T=0\}$. This algebra admits a Cartan decomposition of the form $\mathfrak{sp}_{2n}=\mathfrak l\oplus \mathfrak{p}$ where the first set $(\mathfrak l)$ is the restriction to all skew-symmetric matrices and the second $(\mathfrak p)$ is the restriction to all symmetric matrices \cite{hall2000elementaryintroductiongroupsrepresentations}. This orthogonal decomposition is what gets mapped down to the even-odd split, as we would see later explicitly. To get a flavor of this splitting one can think about the passive versus active components of linear optics. We know that passive elements belong to $O_{sp}$ and hence the Lie algebra of this subgroup is simply contained in $\mathfrak{l}$. This means no component can be present from the orthogonal subspace if the QFI splitting is done correctly. Our QFI splitting leverages this to separate sensing of active vs passive elements. For the rest of the discussion, we would set up the notation as: given a $w\in \mathfrak{sp}_{2n}$, the symmetric part will be denoted as $s\in\mathfrak p$ and the skew-symmetric part as $a\in\mathfrak l$. Then we have the following theorem on how the Bures metric separates into these two components.
\begin{theorem}\label{x3}
    For any Gaussian state given by the representation $(Z,\Gamma)$ we have $ds^2=ds^2_e+ds^2_o$ such that:
    \begin{equation*}
        \begin{split}
            ds^2_e&=\frac{1}{2}\tr[(d\Gamma+[a,\Gamma])(\mathcal M_\Gamma)^{-1}(d\Gamma+[a,\Gamma]) ],\\
            ds^2_o&=\frac{1}{2}\tr[(\{s,\Gamma\})(\mathcal M_\Gamma)^{-1}(\{s,\Gamma\})],
        \end{split}
    \end{equation*}
    where $w:=S^{-1}_Zd(S_Z)\in \mathfrak{sp}_{2n}$.
\end{theorem}
Related to this, we also have a corollary:
\begin{corollary}\label{xx1d}
    For pure states which are of the form $(Z,I/2)$, we have $ds^2_e=0$ and hence $ds^2_{pure}=ds^2_o\propto ds^2_{Siegel}$
\end{corollary}
The above theorem and the corollary show that as long as we are traveling over the manifold of pure states, even contribution to QFI cannot arise. In other words, having a variation in an even part suggests that we have traversed the mixed states too during the sensing process. The expressions for even ($QFI_e$) and odd ($QFI_o$) QFI can be obtained from the Bures metric by replacing $d\Gamma\to \dot \Gamma$ and $w\to S_Z^{-1}\dot{S_Z}$ with an added factor of 4 as per convention. The main proof of the above theorem is provided in the appendix \ref{app1}, but here we just convey the essential steps:  The idea is to combine two facts: (i) the Bures metric is invariant under
Gaussian unitaries (symplectic transformations on the covariance), and (ii) the symplectic Lie algebra admits a Cartan decomposition into \textit{passive} and \textit{active} directions which behave with opposite parity under the symplectic form~$\Omega$ where by even parity of some operator $M$ we mean $P_\Omega(M):=\Omega M\Omega^T=+M$. Likewise, odd parity means $P_\Omega(M)=-M$. First, by symplectic invariance, we are free to work in the frame where the
covariance is $\Gamma$ rather than $V = S_Z \Gamma S_Z^T$. In this frame, the
infinitesimal variation of the covariance induced by a parameter change splits
naturally into two contributions: one coming from changes in the spectral data
$\Gamma$ itself and from passive rotations inside degenerate eigenspaces, and
another coming from genuinely active (squeezing-type) deformations of the
symplectic frame. Using the Cartan decomposition of the symplectic generator
$w := S_Z^{-1} dS_Z$ into a skew-symmetric part $a$ (passive) and a symmetric
part $s$ (active), this yields a decomposition of the tangent vector to the
Gaussian manifold into an \textit{even} part $dv_e$ and an \textit{odd} part $dv_o$. We then fix an \textit{even frame} for~$\Gamma$, chosen so that it commutes with the
symplectic form~$\Omega$. In this frame, the even and odd pieces of the tangent
vector have opposite parity under conjugation by~$\Omega$: the even part
commutes with~$\Omega$, while the odd part anticommutes with~$\Omega$. This
parity defines a $\mathbb Z_2$ grading on the space of symmetric matrices.
The operator that appears in the Bures metric (the superoperator acting on
covariance variations) is built only from $\Gamma$ and $\Omega$, and therefore
commutes with this parity involution. As a result, it is block-diagonal with
respect to the even/odd splitting. When we evaluate the Bures quadratic form on the sum of the even and odd tangent components, the mixed term must vanish because it couples vectors of
opposite parity through a parity-preserving operator. What remains is a sum of
two nonnegative contributions: one supported entirely on the even sector, and
one supported entirely on the odd sector. These are precisely the two pieces
$ds_e^2$ and $ds_o^2$ stated in the theorem. The explicit formulas follow from
writing out the even and odd parts in the $(Z,\Gamma)$ representation.

Theorem \ref{x3} leads to an important consequence in regards to sensing of active versus passive unitary circuits. Physically, if we start from a thermal state and only perform linear interferometers or (passive unitaries) then we always evolve in a way that $QFI_o=0$. In other words, having an odd part suggests that our evolution must have some \textit{active} components that are being used. By definition, passive unitary symplectic matrices are generated by $e^X$ for $X\in \mathfrak l$. These form the subgroup $O_{sp}$. Then we have 
\begin{corollary}
    Let $V(t)=S(t)\Gamma S^T(t)$ and we are sensing $t$ then
    \be
    \begin{split}
    \forall t, S(t)\in O_{sp}&\implies QFI_o=0.
    \end{split}
    \ee
\end{corollary}
This means for states of the form $(iI,\Gamma)$, driving such states via passive optical networks does not produce any odd component of QFI. In general, if $S(t)$ is active, then it is not possible to make any analogous claim. The above splitting is dependent on the frame $S_Z$, but to make general claims and perform computations, it is useful to define another equivalent way of splitting the QFI. Whenever we have a covariance matrix $V$ that is even ($[V,\Omega]=0$), then the corresponding $\mathcal M_V$ commutes with $P_\Omega$, and this will always lead to a splitting of the Bures metric into an even and odd contribution. The main reason for this is $dV\in Sym(2n,R)$ which is space of real symmetric matrices and 
\begin{equation}
    Sym(2n,R)=ker(P_+)\bigoplus ker(P_-),
\end{equation}
where $P_\pm=\frac{1}{2}(X\pm P_\Omega(X))$ are two projectors. The best intuition behind this is to see $P_\Omega$ as a linear operator over the space $Sym(2n,R)$. This operator satisfies $P^2_\Omega(M)=M$ and hence has eigenvalues as $\pm 1$. The above spaces are basically the two eigenspaces corresponding to this decomposition. This leads to the definition of even and odd QFI.\\
\begin{definition}
    [Even and Odd QFI:]
    Given a single parameter variation of Gaussian states $\rho(t)$, we have a variation on the covariance matrix $V(t)=S(t)K(t)S(t)^T$, where $K(t)$ contains the symplectic eigenvalues of $V(t)$. Let $\dot \Sigma(t)=S^{-1}(t)\dot VS^{-T}(t)$ then:
    \begin{equation}
        \begin{split}
            QFI_e(t)&:=2\tr[P_+(\dot \Sigma(t))\mathcal M_{K(t)}^{-1}P_+(\dot \Sigma(t))],\\
             QFI_o(t)&:=2\tr[P_-(\dot \Sigma(t))\mathcal M_{K(t)}^{-1}P_-(\dot \Sigma(t))].\\
        \end{split}
    \end{equation}
    With the original QFI being the sum of two.
\end{definition}
This frame is actually related to the frame we used to split QFI for $(Z,\Gamma)$ via an $O_{sp}$ element. More precisely $K(t)=O(t)\Gamma(t)O^T(t)$ for some $O(t)\in O_{sp}$. In principle, we have defined the split in the Williamson frame in which is even means $[K(t),\Omega]=0$. For any variation of covariance matrix $V(t)$, the Williamson frame always exist and by definition it is even, now from this frame we can travel to other frames via orthogonal symplectic transformations and in each such frame we can have a splitting of QFI because for any even frame $[V,\Omega]=0$, $OVO^T$ is even with $O\in O_{sp}$. Furthermore, we would want $QFI_e$ and $QFI_o$ to be invariant under the choice of any such frames. Above, we had seen two canonical choices, with one being the $\Gamma-$frame and the other being the Williamson frame. Below we show the invariance of our splitting over such frames:\\

\begin{proposition}[ Invariance under even-frame changes]
\label{prop:frame-invariance}
Let $t \mapsto V(t)$ be a differentiable family of centered Gaussian
covariance matrices. For $i=1,2$ let $V_i(t)$
be two \emph{even frames} for $V(t)$ connected by $O(t)\in O_{sp}$, i.e.,
\[
  V(t) = S(t)V_1(t)S(t)^T=S(t)O(t)V_2(t)O(t)^TS(t)^T,
\]
Then the even/odd splitting of the QFI is independent of the
choice of such even frames $V_i$.
\end{proposition}
The above proposition says that to see the splitting, one does not need to fully go to Williamson's frame by undoing the symplectic evolution; rather, we can always keep a gauge part of the orthogonal symplectic intact. The proof of this is deferred to appendix \ref{app1}. The basic intuition is that within the same trajectory of evolution of $V(t)$, every even frame connected to one another via $O(t)$ element is equivalent. 

\subsection{Explicit expressions in Williamson frame}
Based on the above discussions, we see that overall, the even frames connected by orthogonal symplectic matrices, our split is well defined, and each term (even or odd) is itself invariant. We can use this to evaluate these contributions explicitly. $\dot\Sigma\in Sym(2n,R)$ means we can break it into even and odd components explicitly as:
\begin{equation}\label{thanos2}
    \dot\Sigma=\begin{pmatrix} M& -N\\
N& M\end{pmatrix}+\begin{pmatrix} A & B\\
B& -A\end{pmatrix},
\end{equation}
where $M=M^T, N=-N^T, A=A^T, B=B^T$. Here, the first matrix is always even and the other odd. This is the decomposition of $\dot\Sigma$ into even and odd terms. One can explicitly check that the first part of the above decomposition commutes with the symplectic form $\Omega$ and the second part anti-commutes. Further, it exhausts all possible symmetric matrices over our desired space. Under this decomposition, we have the following theorem:
\begin{theorem}\label{will decom}
    Let $(k_1,\ldots,k_n)$ be the symplectic eigenvalues. Define $\alpha^{\pm}_{ij}=4k_ik_j\pm1$ then:
    \begin{equation}
    \begin{split}
QFI_e(t)&=4\sum_{ij}\frac{M^2_{ij}+N^2_{ij}}{\alpha^{-}_{ij}},\\
QFI_o(t)&=4\sum_{ij}\frac{A^2_{ij}+B^2_{ij}}{\alpha^{+}_{ij}}.
\end{split}
    \end{equation}
\end{theorem}
It is important to stress that everything has time-dependence, where the symplectic values and the decomposition changes with time, so this equation is evaluated at time $t$. More precisely, given the evolution $V(t)$, we extract the $K(t)=diag(k_1(t),\ldots,k_n(t))$ information, which is just Williamson's decomposition. Following this we define $\alpha^{\pm}_{ij}(t)$ and the matrices $M(t),N(t),A(t),B(t)$ comes from decomposing the velocity term $\dot\Sigma(t)$. The proof is provided in appendix \ref{app1}. The key simplification of the Williamson frame is that, because $V$ is diagonal in mode indices, the superoperator $M_V = 4L_V + L_\Omega$
acts \emph{diagonally} on matrix elements when expressed in a suitable
orthonormal basis of symmetric or anti-symmetric matrices. More concretely, for each
pair of mode indices $(i,j)$ one can choose basis elements supported
only on the $(i,j)$ and $(j,i)$ entries in $M,N,A,B$, and in this basis
$\mathcal M_V$ acts by simple multiplication with scalar factors
\[
\alpha^-_{ij} = 4k_i k_j - 1\text{(even sector)},\quad
\alpha^+_{ij} = 4k_i k_j + 1\text{(odd sector)}.
\]
Thus, on each $(i,j)$ block, $M_V$ is already diagonal, and its
pseudo-inverse $M_V^{-1}$ is obtained by dividing the corresponding
components of $M,N,A,B$ by $\alpha^\mp_{ij}$. A closer look at the above theorem \ref{will decom} shows that the $QFI_e(t)$ or $QFI_o(t)$ are just appropriately weighted Frobenius norm of even tangent vector $\dot\Sigma_e(t)= P_+(\dot\Sigma)$ and the odd tangent vector $\dot\Sigma_o(t)= P_-(\dot\Sigma)$ respectively. The weightage is exactly the inverse of the $\alpha^\mp_{ij}$ factor, which captures the noise term from the state traveling into mixed states. The $QFI_e(t)$ blows up whenever the state approaches pure states, even if $QFI_o(t)$ is finite. This discontinuity is the same as that discussed in \cite{Safranek:2016blj}. If the dynamics are restricted to the pure state manifold, then $QFI_e(t)$ vanishes because of the use of pseudo-inverse, and the contribution only comes from the odd term, which agrees with our corollary \ref{xx1d}. The above idea can easily be extended to multi-parameter estimation:
\begin{corollary}\label{loss2}
    If we are estimating $\vec\theta=(\theta_1,\ldots,\theta_n)$ then, every velocity term $\dot \Sigma_a=\frac{\partial \Sigma}{\partial \theta_a} $ can be appropriately decomposed and we have
    \begin{equation}
        \begin{split}
(QFI_e)_{ab}&=4\sum_{ij}\frac{M^a_{ij}M^b_{ij}+N^a_{ij}N^b_{ij}}{\alpha^{-}_{ij}},\\
(QFI_o)_{ab}&=4\sum_{ij}\frac{A^a_{ij}A^b_{ij}+B^a_{ij}B^b_{ij}}{\alpha^{+}_{ij}},
        \end{split}
    \end{equation}
    with the sum being the entire QFI matrix.
\end{corollary}
The split of even and odd parts also brings interesting, straightforward consequences for thermometry \cite{cenni,Mehboudi_2019}. We call a parameter $t$ thermometric if it brings about a change only in the symplectic eigenvalues. This automatically means that variations along the odd split die off. Then the equation above becomes $P=\dot K$, and all other matrices are 0 in this special case of thermometric parameter, which gives:
\begin{corollary}
    For any thermometric parameter $t$ we have:
    \begin{equation}
        QFI(t)=QFI_e(t)=4\sum_i\frac{\dot k^2_i}{4k^2_i-1},
    \end{equation}
\end{corollary}

This relation is of a similar form to that obtained in \cite{cenni}. For single mode, if the state is at thermal equilibrium, then $k(T)=\frac{1}{2}\coth{\frac{\omega}{2T}}$ where $\omega$ is the energy of the Bosonic mode and $T$ is the temperature. Plugging this into the above equation gives:
\begin{equation}
    QFI(T)=QFI_e(T)=\frac{\omega^2}{4T^2\sinh^2(\frac{\omega}{2T})},
\end{equation}
which matches the standard expression in \cite{Correa_2015}. Related to this, we can also give a state-dependent lower bound on $QFI_e(t)$ in terms of the rate of change of purity of the state. Let purity be defined as $\mu=\frac{1}{\sqrt{det(V)}}$ (\cite{Serafini2017QuantumCV} upto some rescaling which we ignore) then:
\begin{theorem}\label{thanos3}
    Given variation of covariance matrix $V(t)$ let purity be $\mu(t)$ then,
    \be
    QFI_e(t)\geq \frac{8}{8n-||V(t)^{-1}||^2}\{\frac{d(\ln{\mu(t)})}{dt}\}^2.
    \ee
\end{theorem}

The norm used above is the Hilbert–Schmidt norm. So, given a fixed amount of even QFI ( or equivalently, total QFI), there is a bound on how fast global mixedness can change with the parameter. This is a Gaussian, even-sector analogue of \textit{quantum speed limits}: the Bures metric controls how fast certain scalar functionals (here, purity) can change along a curve in state space. In the Gaussian setting, the even QFI hence quantifies not only the sensitivity of populations, but also acts as a speed limit for how fast a parameter can change the global mixedness (purity) of the state. Conversely, measuring how rapidly the purity changes with the parameter provides a simple experimentally accessible lower bound on the even-sector QFI. The bound also has an interesting consequence in relation to how $QFI_e(t)$ diverges as our state leaves the pure state manifold $(t=0)$. For this consider a linear expansion of symplectic values $k_i(t)=\frac{1}{2}+\epsilon t+\mathcal O(t^2)$, then the first lower bound diverges as $t^{-1}$ and hence, in these places $QFI_e(t)$ blows up. As one approaches pure states, the metric in the even directions becomes very \textit{stiff} (since spectral deformations are heavily constrained by uncertainty), and the bound reflects that stiffness. 
\section{\label{sec:level5}Examples \protect}
Here, we discuss several examples of how the split is helpful or provides useful information related to the sensing parameters. We start with a very simple example of sensing the beam-splitter angle to show the application of how QFI can be deduced from graphs. We start with two modes, one squeezed and the other anti-squeezed send them via a beam-splitter
\be
\rho(r,t)=\hat{U}_{BS}(t/2)\hat S_1(r)\otimes \hat S_2(-r)[\ket{0}\bra{0}].
\ee
Here, the notation $\hat S(r)$ is the squeezing unitary operator with the symplectic representation as $S(r)=diag(e^r,e^{-r})$. The beam-splitter unitary is denoted by $\hat U_{BS}(t)$ whose symplectic can be defined as $S_{BS}(t/2)= e^{t\frac{G_B}{2}}$ where $G_B$ generates the symplectic of beam-splitter and is given by:
\begin{equation}
    \begin{split}
        G_B&=\begin{pmatrix} R& 0\\
0& R\end{pmatrix},
    \end{split}
\end{equation}
where $R=\begin{pmatrix} 0& 1\\
-1& 0\end{pmatrix}$. As per the figure below, we find the output graph as
\begin{figure}[h]
\begin{centering}
\includegraphics[scale=0.5]{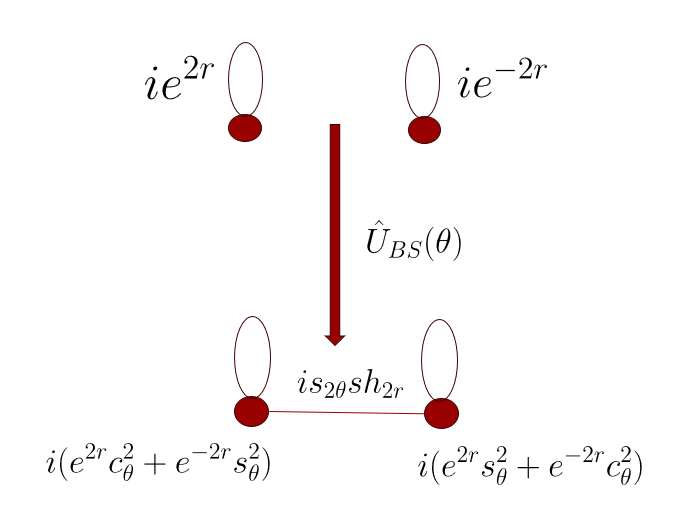}
\caption{\label{fig:epsart} Graphical representation of how the beam-splitter changes and distributes correlations. Here $c_\theta=cos(\theta)$ and $sh(r)=sinh(r)$.}
\end{centering}
\end{figure}
\be
Z=i\begin{pmatrix}
    e^{2r}c^2_{t/2}+e^{-2r}s^2_{t/2} & s_tsh_{2r}\\
    s_tsh_{2r} & e^{-2r}c^2_{t/2}+e^{2r}s^2_{t/2},
\end{pmatrix}
\ee
and this gives the derivative as:
\be
\frac{dZ}{dt}=i sh_{2r}\begin{pmatrix}
    -s_{t} & c_{t}\\
    c_t & s_{t}
\end{pmatrix}.
\ee
Now, using the equation for QFI in terms of graphical parameters, we get
\be
QFI_o(t)=QFI(t)=\frac{1}{2}\tr[Y^{-1}\frac{d(Z)}{dt}Y^{-1}\frac{dZ^*}{dt}]=sh^2_{2r}.
\ee
Here, $Z=iY$ and the entire graph is imaginary. 
For $t=0$, the $Z=i\times diag(e^{2r},e^{-2r})$ is just two uncorrelated squeezed vacuum states. Similarly, for $t=\pi/2$, we get the two-mode squeezed vacuum (TMSV). As sensing is over pure states, our even contribution vanishes, and the entire QFI is the odd contribution. This goes to $0$ as $r$ goes to $0$, otherwise we get an exponential metrological gain $QFI\approx \frac{1}{4}e^{4r}$.

\subsection{Sensing temperature difference}
Let us do an explicit 2-mode example: We will start with thermal states in both modes such that the initial covariance matrix is $diag(K,K)$ with $K=diag(v_1,v_2)$. Then we apply a squeezing operation on the first mode and an anti-squeezing on the second, followed by a beam-splitter operation with time-dependent angle $t$:
\begin{equation}
    \rho(t,r)=\hat U_{BS}(t/2)\hat S_1(r)\otimes \hat S_2(-r)[\rho_{th}].
\end{equation}

The idea is to sense $t$ with thermal states as a probe instead of the vacuum. After one mode squeezing operations, let $V_{th}\to V(r)$ (where $V_{th}$ is the covariance matrix of $\rho_{th}$), then we perform the beam-splitter operation to get:
\begin{equation}
    V(r,t)=S_{BS}(t/2)V(r)S_{BS}^T(t/2),
\end{equation}
the derivative $\dot V|_{t=0}(r)$ is given by:
\begin{equation}
   \dot V|_{t=0}(r)=\frac{1}{2}[G_B,V(r)],
\end{equation}
where we have expanded the exponential from $S_{BS}(t/2)= e^{t\frac{G_B}{2}} $ and used $G_B^T=-G_B$. But we need to push it back to the Williamson frame, which at $t=0$ is just the back action of $S_1(r)\oplus S_2(-r)$. Overall, it means that we need the even-odd split for 
\begin{equation}
    \dot\Sigma|_{t=0}=[S_1(-r)\oplus S_2(r)](\dot V|_{t=0})[S_1(-r)\oplus S_2(r)]
\end{equation}
Performing the multiplications and decomposition according to Eq. \eqref{thanos2} gives: $M=\frac{(v_2-v_1)\cosh(2r)}{2}\sigma_x, N=0, A=\frac{-(v_2+v_1)\sinh(2r)}{2}\sigma_x, B=0$
and it results to:
\begin{equation}
    \begin{split}
        QFI_e(v_1,v_2,r)&=\frac{2\cosh^2(2r)(v_2-v_1)^2}{4v_1v_2-1},\\
        QFI_o(v_1,v_2,r)&=\frac{2\sinh^2(2r)(v_2+v_1)^2}{4v_1v_2+1}.
    \end{split}
\end{equation}
Based on the above equations, we can draw several physical interpretations of the splitting:

\paragraph*{(i) Thermal-contrast interferometry-even sector:}
The even contribution $QFI_e$ quantifies the sensitivity arising from
\emph{population asymmetry} between the two thermal modes. It is proportional
to the squared difference of the symplectic eigenvalues,
\begin{equation}
  QFI_e \propto (v_2 - v_1)^2,
\end{equation}
and vanishes whenever $v_1 = v_2$. In the unsqueezed case $r=0$, we obtain
\begin{equation}
  QFI_e(v_1,v_2,0)
  = \frac{2(v_2 - v_1)^2}{4v_1 v_2 - 1},\qquad
  QFI_o(v_1,v_2,0) = 0,
\end{equation}
so a bare beam splitter on two thermal inputs is sensitive to the mixing
angle $t$ \emph{only} if there is a temperature (population) difference
between the modes. The denominator $4v_1 v_2 - 1$ plays the role of a
global-noise penalty: larger thermal occupancies in both modes suppress the information per unit temperature contrast. Thus, in the absence of squeezing, this setup behaves as a \emph{thermal-contrast interferometer}: the parameter $t$ is sensed purely through spectral (even-sector) information carried by $(v_2 - v_1)$, and the odd sector is completely silent.
\paragraph*{(ii) Correlation-based sensing-odd sector:}
The odd contribution $\mathrm{QFI}_o$ quantifies sensitivity arising from
\emph{correlation-generating} dynamics. It scales with the sum,
\begin{equation}
  QFI_o \propto (v_1 + v_2)^2,
\end{equation}
and is activated as soon as the local squeezing is nonzero. In particular,
when $v_1 = v_2 = v$ we find
\begin{equation}
  QFI_e(v,v,r) = 0,\qquad
  QFI_o(v,v,r)
  = 2\sinh^2(2r)\,\frac{(2v)^2}{4v^2 + 1}.
\end{equation}
In this regime, the two modes start at the \emph{same} temperature, so there is no spectral contrast for the beam splitter to exploit. Nevertheless, the
local squeezers $S_1(r)$ and $S_2(-r)$ prepare the modes with opposite
quadrature anisotropies: one mode is squeezed in a given quadrature while
the other is anti-squeezed. The subsequent beam splitter then converts this
anisotropy into genuine intermode correlations (and, in the pure limit,
entanglement). The output covariance depends on $t$ solely through its
correlation structure, while the symplectic eigenvalues remain unchanged.
Accordingly, $QFI_e$ vanishes and the entire metrological power
comes from the odd sector $QFI_o$. The scaling
\begin{equation}
  QFI_o \sim \sinh^2(2r)\,(v_1 + v_2)^2,
\end{equation}
has a clear interpretation: $\sinh^2(2r)$ quantifies the strength of the
correlation-generating (squeezing) resource, while $(v_1 + v_2)^2$ shows
that, even in the equal-temperature case, a larger total photon number in the two modes enhances the correlation-based sensitivity.

\paragraph*{(iii) Limiting cases and consistency.}
Several limits provide useful checks:
\begin{itemize}
  \item \emph{No sensitivity:} if $v_1=v_2$ and $r=0$, then
        $QFI_e = QFI_o = 0$, which reflects the fact that a beam splitter acting on two identical thermal modes generates no
        $t$-dependence in the state.
  \item \emph{Pure limit:} if $v_1=v_2 = \tfrac12$ (both inputs in vacuum),
        then $QFI_e = 0$ and
        \[
          QFI_o \to \sinh^2(2r),
        \]
        in agreement with the pure-graph example and with the general result
        that the even sector vanishes on the pure manifold.
\end{itemize}
\begin{figure}[h]
\begin{centering}
\includegraphics[scale=0.5]{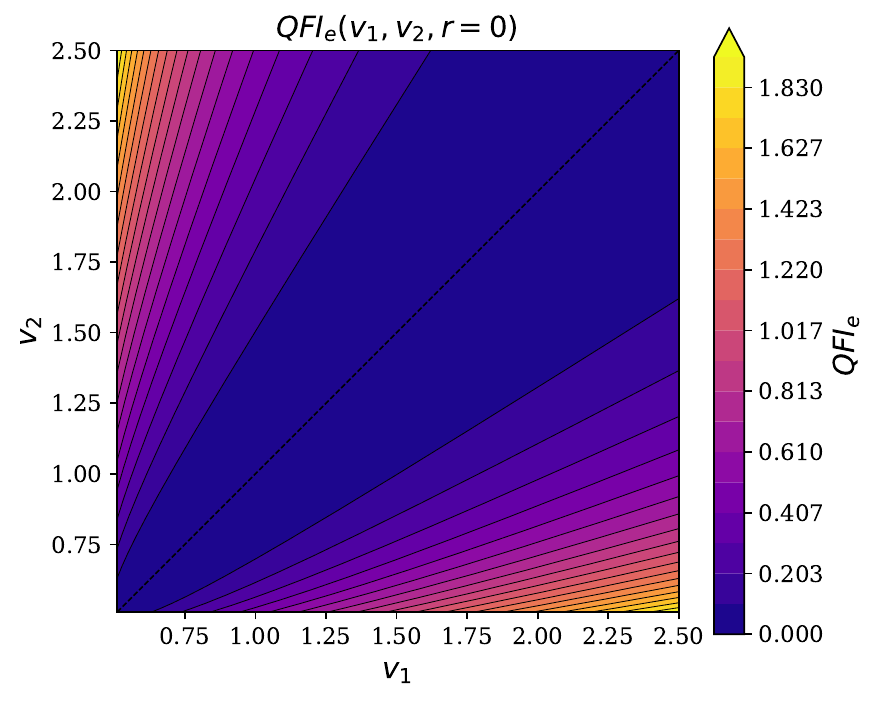}
\label{fig: epsart}
\caption{\justifying{\label{fig:epsart} Plot of the even quantum Fisher information $QFI_e(v_1,v_2,r=0)$.
It is zero when $v_1 = v_2$ and hence can witness a temperature difference between modes.
For very asymmetric mode temperatures (even for small differences), the corners of the plot light up.
}}
\end{centering}
\end{figure}
\begin{figure}[h]
\begin{centering}
\includegraphics[scale=0.5]{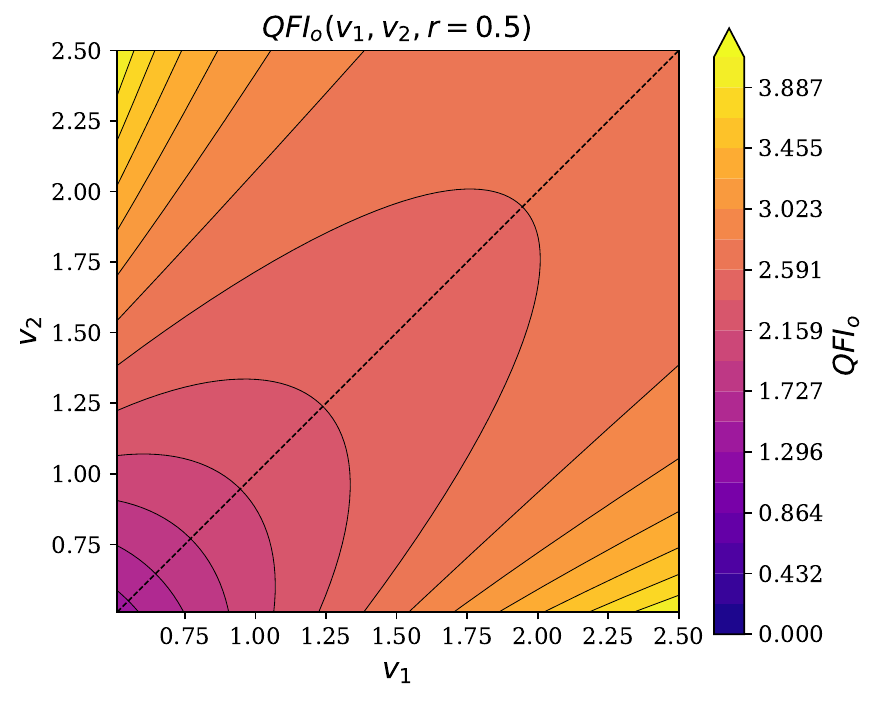}
\caption{\justifying{\label{fig:epsart}. Plot of the odd quantum Fisher information $QFI_o(v_1,v_2,r=0.5)$.
        Unlike its even counterpart, this contribution never vanishes and increases more rapidly with growing asymmetry.
        A large odd component reflects that the parameter is probed along a direction that enhances inter-mode correlations, leading to an increased sensitivity of the state in this sector.}}
\end{centering}
\end{figure}

\subsection{Sensing of transmissivity of loss channel}
We start with a vacuum state $\ket{0}$. Then we pass it through a squeezer followed by a loss channel. The loss channel is a Gaussian completely positive trace-preserving map (CPTP map~\cite{niel}). Such a channel can be represented by noise matrices $(X,Y)$ such that under the channel the covariance matrix transforms as $V\to XVX^T+Y$~\cite{Serafini2017QuantumCV}. For the loss channel (denoted by $\mathcal L_\eta$) we have $X=\sqrt{\eta}I$ and $Y=\frac{(1-\eta)}{2}I$. The entire evolution is given as:
\begin{equation}
  \rho(\eta,r)= \mathcal L_\eta( \hat S_1(r)(\ket{0}\bra{0})\hat S_1(r)^\dagger).
\end{equation}
We now plan to sense $\eta$. For the exact calculation, we start from the calculation of the covariance matrix of $\rho(\eta,r)$
\begin{equation}\label{loss1}
    V(\eta,r)=\begin{pmatrix} \Lambda_1&0 \\
0& \Lambda_2\end{pmatrix},
\end{equation}
where $\Lambda_1=(1-\eta)/2+\eta e^{2r}/2$ and $\Lambda_2=(1-\eta)/2+\eta e^{-2r}/2$. To calculate the even and odd QFI, we need the dynamic frame change symplectic matrix $S_{\eta,r}$ such that $K(\eta,r)=S^{-1}_{\eta,r}V(\eta,r)S^{-T}_{\eta,r}$ gives the Williamson form. The exact form of $S_{\eta,r}$ is given by
\begin{equation}
    S_{\eta,r}=diag \left(\left(\frac{\Lambda_1}{\Lambda_2}\right)^{1/4}, \left(\frac{\Lambda_2}{\Lambda_1}\right)^{1/4}\right),
\end{equation}
and $K(\eta,r)=diag(\sqrt{\Lambda_1\Lambda_2},\sqrt{\Lambda_1\Lambda_2})$. The velocity term in Williamson frame is given by $\dot\Sigma(\eta,r)=S^{-1}_{\eta,r}\dot V S^{-T}_{\eta,r}$ which can be written explicitly in terms of $\Lambda_1$ and $\Lambda_2$ as
\begin{equation}
    \dot\Sigma(\eta,r)=\frac{1}{2}diag((e^{2r}-1)(\frac{\Lambda_2}{\Lambda_1})^{1/2},(e^{-2r}-1)(\frac{\Lambda_1}{\Lambda_2})^{1/2}).
\end{equation}
Using the decomposition into even and odd sectors gives
\begin{equation}\label{loss3}
\begin{split}
    QFI_e&=\frac{4m^2}{4{\Lambda_1\Lambda_2}-1},\\
    QFI_o&=\frac{4a^2}{4{\Lambda_1\Lambda_2}+1},
\end{split}
\end{equation}
where $m=\frac{d_1+d_2}{2}$ and $a=\frac{d_1-d_2}{2}$. Here, we have set $d_1=\dot\Sigma(\eta,r)_{00}$ and $d_2=\dot\Sigma(\eta,r)_{11}$. Similar to the case of temperature sensing, we can give physical interpretations of different regimes and how $QFI_e$ and $QFI_o$ behave in these regimes:

\paragraph*{(i) Highly lossy channel, $\eta\to 0$:} As the transmissivity goes to 0, we always get the vacuum state as output irrespective of our input state. In this regime for moderate squeezing $0\leq r\leq 4$ values, we have the following behavior of the QFI:
\begin{align}
  QFI_e(\eta,r)
  &\simeq
   \frac{\cosh 2r - 1}{2\eta},
   \label{eq:qfi-e-eta0}\\[0.3em]
  QFI_o(\eta,r)
  &\simeq  \frac{\sinh^2 2r}{2}.
\end{align}
Thus, near $\eta=0$, the even sector diverges like \(QFI_e \sim (\cosh 2r -1)/(2\eta)\), while the odd sector tends to a finite value of order $\sinh^2 2r$. Physically, for $\eta\approx 0$ the channel almost completely discards the input and outputs vacuum; nevertheless, \emph{any} small increase in $\eta$ introduces a finite amount of squeezed energy, drastically changing
the symplectic eigenvalues and hence the purity. The even sector captures this \textit{purity sensitivity} and therefore blows up as $1/\eta$. 

\paragraph*{(ii) Near-identity channel, $\eta \to 1$:}

At $\eta = 1$, the channel becomes the identity and the output is a pure
squeezed vacuum. In this regime, for moderate squeezing $0 \leq r \leq 4$, we obtain the following behaviour of the QFI:
\begin{align}
  QFI_e(\eta,r)
  &\simeq \frac{\cosh 2r - 1}{2(1-\eta)}, \\
  QFI_o(\eta,r)
  &\simeq \frac{\sinh^2 2r}{2}.
\end{align}
Thus, the behaviour near $\eta = 1$ is completely analogous to the limit
$\eta \to 0$: the even sector diverges as
\(
  QFI_e \sim (\cosh 2r - 1)/[2(1-\eta)]
\),
while the odd sector remains finite. Physically, at $\eta = 1$, the output is a pure squeezed state. Any small decrease in $\eta$ (small amount of loss) produces a finite first-order change in the symplectic eigenvalue and hence in the purity. The statistical metric along this
\textit{purity-sensitive} direction is singular at the pure-state boundary, so the QFI for the transmissivity diverges. This is purely an even-sector effect: the odd sector detects only squeezing-shape changes at fixed symplectic eigenvalue, which remain smooth at $\eta = 1$.

\paragraph*{(iii) Moderate loss, e.g.\ $\eta \approx 0.5$:}

For intermediate transmissivities $\eta \in (0,1)$ away from the pure boundaries, the symplectic eigenvalue is strictly larger than $1/2$. In this case, there is no divergence, and both $QFI_e(\eta,r)$ and $QFI_o(\eta,r)$ remain finite. The state is strongly mixed for typical
values of $r$, so the purity changes only slowly with $\eta$ and the even-sector QFI is correspondingly suppressed. The structure of the QFI in this regime is conveniently summarized by introducing the odd QFI fraction
\begin{equation}
  F_{\mathrm{odd}}(\eta,r)
  = \frac{QFI_o(\eta,r)}{QFI_e(\eta,r) + QFI_o(\eta,r)},
\end{equation}
and the output purity is
\begin{equation}
  \mu(\eta,r)
  = \frac{1}{2\sqrt{\det V(\eta,r)}},
\end{equation}
where $V(\eta,r)$ is the covariance matrix of the output state.

Figure~\ref{fig:loss-Fodd} shows $F_{\mathrm{odd}}(\eta,r)$ as a
function of the channel transmissivity $\eta$ for a fixed range of
squeezing values $r$, while Fig.~\ref{fig:loss-purity} displays the
corresponding purity $\mu(\eta,r)$. One clearly sees that in regions
where the purity varies weakly with $\eta$, the odd sector
dominates the total QFI ($F_{\mathrm{odd}} \approx 1$), whereas in
regions where the purity is highly sensitive to $\eta$, the even contribution carries most of the total QFI.

\begin{figure}[h]
  \centering
  \includegraphics[scale=0.5]{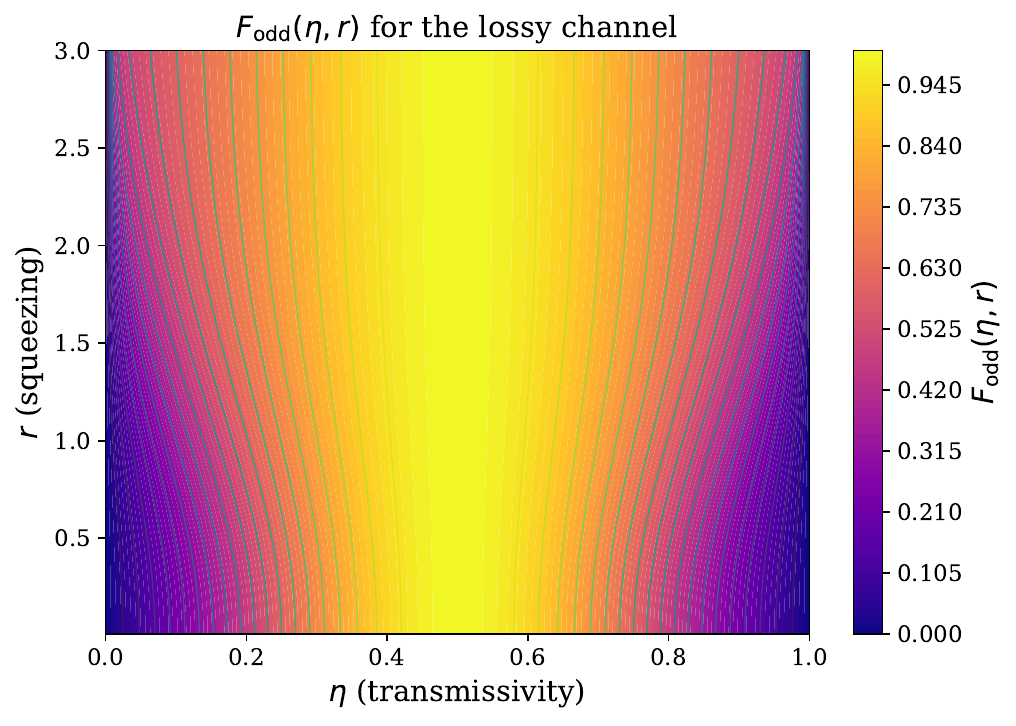}%
  \caption{\justifying{
    Odd QFI fraction
    $F_{\mathrm{odd}}(\eta,r)
      = QFI_o(\eta,r)/[QFI_e(\eta,r)+QFI_o(\eta,r)]$
    for the lossy channel, plotted as a function of the transmissivity $\eta$ (horizontal axis) for a representative
    choice of input squeezing $r$ (see main text for parameters).
    The plot illustrates that at moderate losses the odd contribution
    dominates the total QFI, whereas near the extreme-loss and
    near-identity limits the even sector becomes increasingly important.}
  }
  \label{fig:loss-Fodd}
\end{figure}

\begin{figure}[h]
  \centering
  \includegraphics[scale=0.5]{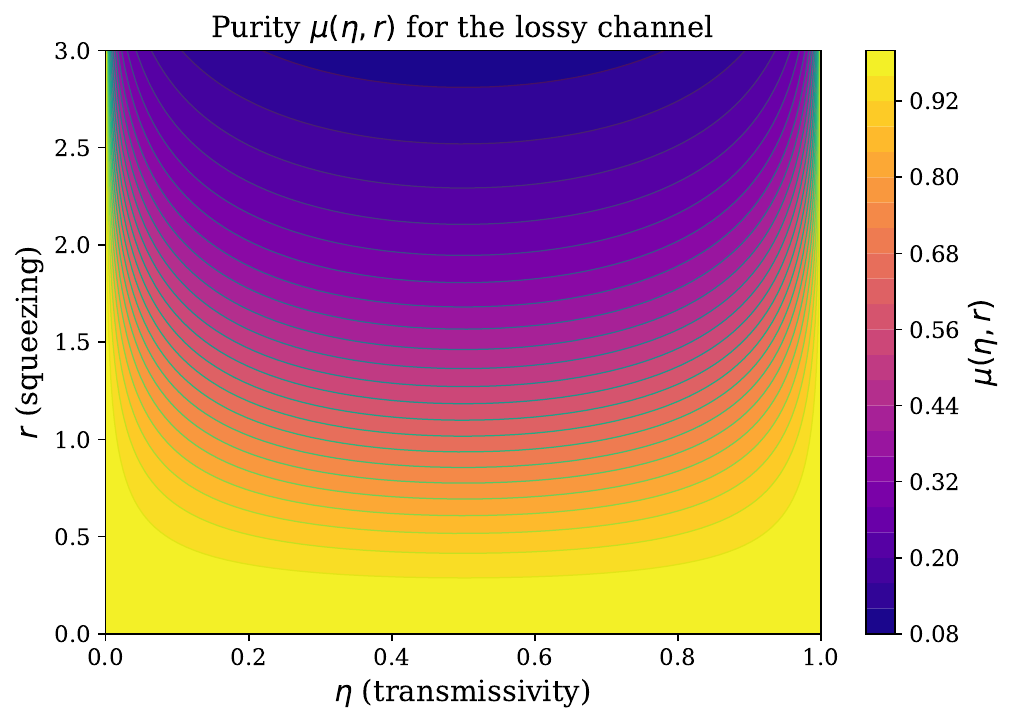}%
  \caption{\justifying{
    Purity of the output state,
    $\mu(\eta,r) = 1/[2\sqrt{\det V(\eta,r)}]$,
    as a function of the transmissivity $\eta$ (horizontal axis)
    for the same range of $r$ as in Fig.~\ref{fig:loss-Fodd}.
    Comparing with Fig.~\ref{fig:loss-Fodd} shows that regions where
    $\mu(\eta,r)$ changes slowly with $\eta$ correspond to a dominant
    odd contribution to the QFI, whereas regions with a rapidly varying
    purity are associated with a large even-sector contribution.
  }}
  \label{fig:loss-purity}
\end{figure}
The same analysis can be repeated for the phase-insensitive
amplification channel, now with
\begin{equation}
  \Lambda_1 = \frac{g - 1}{2} + \frac{g\,e^{2r}}{2},
  \qquad
  \Lambda_2 = \frac{g - 1}{2} + \frac{g\,e^{-2r}}{2},
\end{equation}
where $g$ is the gain and $r$ is the input squeezing. We again consider
the odd QFI fraction $F_{\mathrm{odd}}(g,r)$ and the purity
$\mu(g,r)$, now as functions of the gain $g$. As in the lossy case, the
plots reveal that if the state lies on (or close to) the manifold of pure states, but the rate of change of purity with respect to the sensed parameter is large, then the even contribution dominates the QFI. In contrast, when the purity changes only weakly with the parameter, the odd sector provides the leading contribution.

Moreover, asymmetries in the purity landscape as a function of $g$
are directly reflected in corresponding asymmetries of the odd QFI fraction. This supports the idea that even and odd contributions probe, respectively, purity-changing and
correlation changing directions in parameter space. In the next section, we show how the full quantum Fisher information matrix (QFIM)
can be decomposed into even and odd parts along these two complementary directions.

\begin{figure}[h]
  \centering
  \includegraphics[scale=0.5]{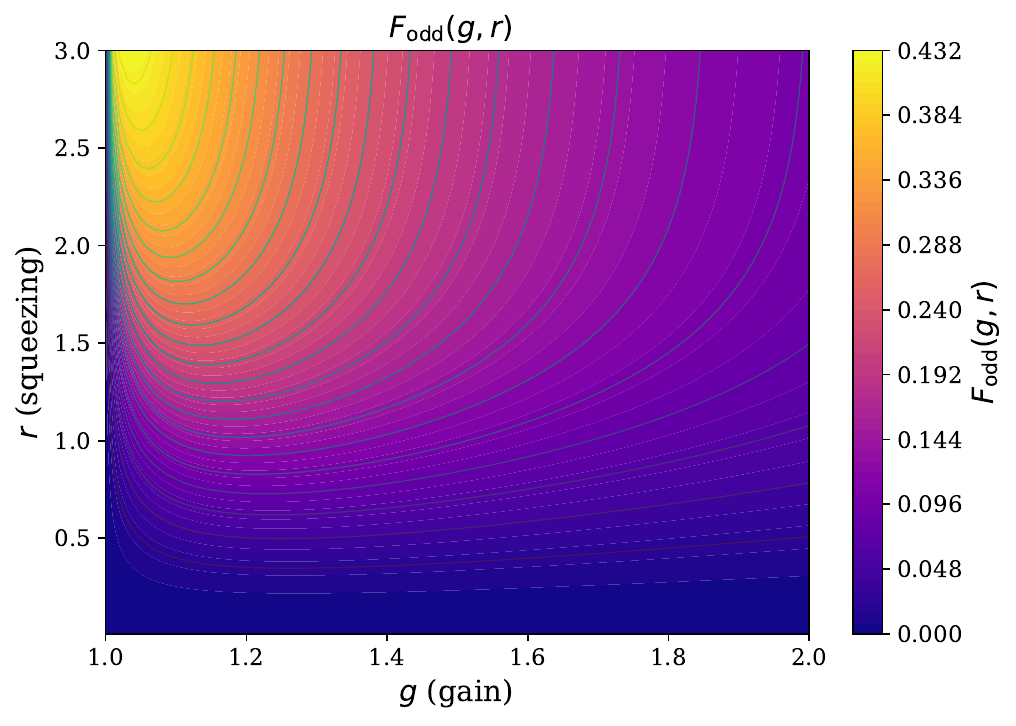}%
  \caption{\justifying{
    Odd QFI fraction
    $F_{\mathrm{odd}}(g,r)
      = QFI_o(g,r)/[QFI_e(g,r)+QFI_o(g,r)]$
    for sensing the gain $g$ of the amplification channel when the probe is a squeezed state. The horizontal axis shows the gain $g$; the odd contribution dominates in regions where the purity changes slowly with $g$, while the even sector takes over when the state approaches purity-sensitive boundaries.
  }}
  \label{fig:amp-Fodd}
\end{figure}

\begin{figure}[h]
  \centering
  \includegraphics[scale=0.5]{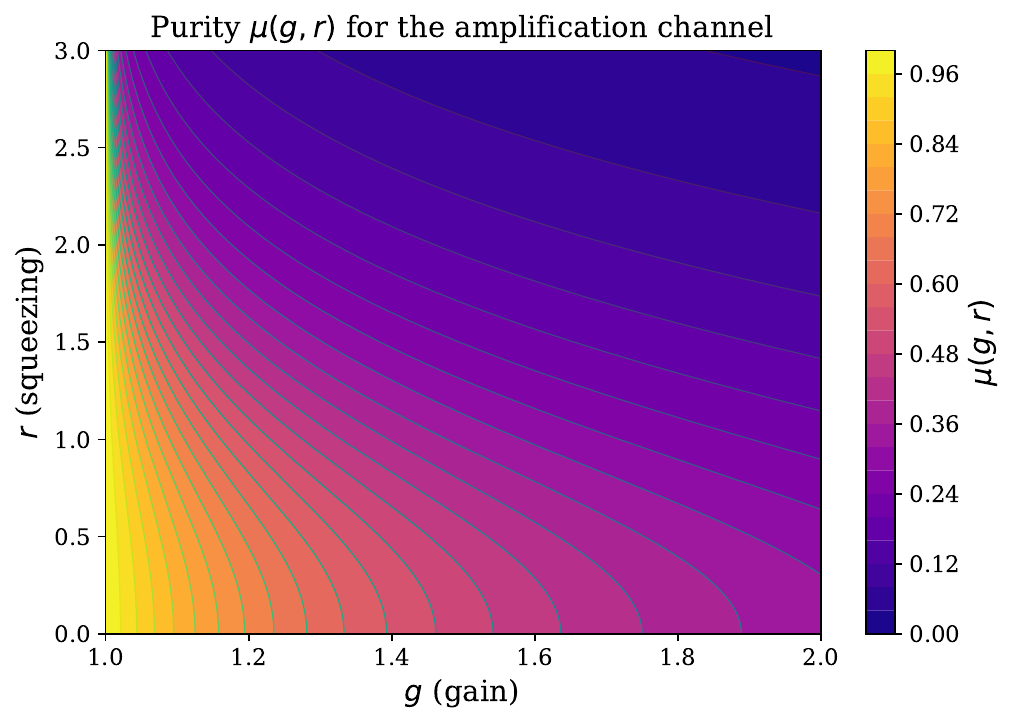}%
  \caption{\justifying{
    Purity of the output state for the amplification channel,
    $\mu(g,r) = 1/[2\sqrt{\det V(g,r)}]$, plotted as a function of the
    gain $g$ (horizontal axis) for the same squeezed input as in
    Fig.~\ref{fig:amp-Fodd}.
    Regions where the purity varies slowly with $g$ correlate with a
    dominant odd QFI fraction in Fig.~\ref{fig:amp-Fodd}, whereas
    regions of rapidly changing purity are associated with a dominant
    even contribution to the QFI.
  }}
  \label{fig:amp-purity}
\end{figure}

\subsection{Quantum Phase estimation and QFIM}

For this, we consider the estimation of the $\eta$ and $\theta$ for the state:
\begin{equation}
  \rho(\eta,r,\theta)= \mathcal{R}_\theta\mathcal L_\eta( \hat S_1(r)(\ket{0}\bra{0})\hat S_1(r)^\dagger)\mathcal R_\theta^\dagger,
\end{equation}
where $\mathcal R_\theta$ is the phase shift unitary whose symplectic is defined as $R_\theta=e^{\theta G_P}$ with
\be
G_P=\begin{pmatrix} 0 & 1\\
-1& 0\end{pmatrix}.
\ee
It is easy to see that the frame change will be carried out by the symplectic $S_F=S^{-1}_{\eta,r}e^{-\theta G_P}$. This gives $\dot\Sigma_{\eta}=S_F(\dot V_\eta)S_F^T$ and $\dot\Sigma_{\theta}=S_F(\dot V_\theta)S_F^T$. Here, $\dot V_\eta=\partial_\eta V$ and $\dot V_\theta=\partial_\theta V$ where $V$ is the entire covariance matrix. The covariance matrix after the loss channel and before the phase is diagonal and given by Eq. \eqref{loss1}, and the phase rotation acts by congruence:
\begin{equation}
  V
  = R_\theta\,V(\eta,r)\,R_\theta^T.
\end{equation}

From this we can extract the $M,N,A,B$ matrices (they are just numbers in this case) as:
\be
\begin{split}
    &M^\eta= m, A^\eta=a, B^\eta=N^\eta=0,\\
    &B^\theta=\Lambda_2-\Lambda_1=-\eta \sinh(2r), A^\theta= M^\theta=N^\theta=0.
\end{split}
\ee
Here, $m,a$ are already defined in Eq. \eqref{loss3}. This shows that the off-diagonal terms of QFI vanish using our formula. It means that the two parameters are orthogonal. Overall, this gives two parts of QFIM:
\be
\begin{split}
    QFIM_o&=\begin{pmatrix}
     \frac{4a^2}{4{\Lambda_1\Lambda_2}+1}& 0,\\
    0&\frac{4\eta^2\sinh^2(2r)}{4{\Lambda_1\Lambda_2}+1}
    \end{pmatrix}\\
    QFIM_e&=\begin{pmatrix}
    \frac{4m^2}{4{\Lambda_1\Lambda_2}-1}& 0\\
    0&0
\end{pmatrix}.
\end{split}
\ee
Basically, we used corollary \ref{loss2} for this specialized case.
To visualize the trade-off between these QFIM sectors of even and odd, we can define vectors $v_\lambda=(\sqrt{QFI_e^\lambda},\sqrt{QFI_o^\lambda})$ where $\lambda$ is the estimation parameter. For the above example we have $v_\eta=(\frac{4m^2}{4{\Lambda_1\Lambda_2}-1}, \frac{4a^2}{4{\Lambda_1\Lambda_2}+1} )$ and $v_\theta=(0,\frac{4\eta^2\sinh^2(2r)}{4{\Lambda_1\Lambda_2}+1})$. We plot these vectors in a plane as a function of $\eta$. Because of the symmetry of QFI for sensing of $\eta$ around $\eta=0.5$ we modify $v_\eta=(\frac{sgn(\eta-0.5)4m^2}{4{\Lambda_1\Lambda_2}-1}, \frac{4a^2}{4{\Lambda_1\Lambda_2}+1} )$. Here, $sgn(x)$ is the sign function that takes the value $+1$ for $x> 0 $, $-1$ for $x<0$ and $0$ at $x=0$. These vectors basically separate the even and the odd components of QFI and provide a visual interpretation, that as we vary $\eta $, our sensing uses resources related to spectral changes as well as those that squeeze or change the shape of our state. This can be seen from the fact that we trace a curve in the entire plane and not just along a constraint axis. On the other hand, the QFI for sensing the $\theta$ parameter is only constrained in the y-axis of the plot because it gives only an odd contribution.

\begin{figure}[h]
\begin{centering}
\includegraphics[scale=0.4]{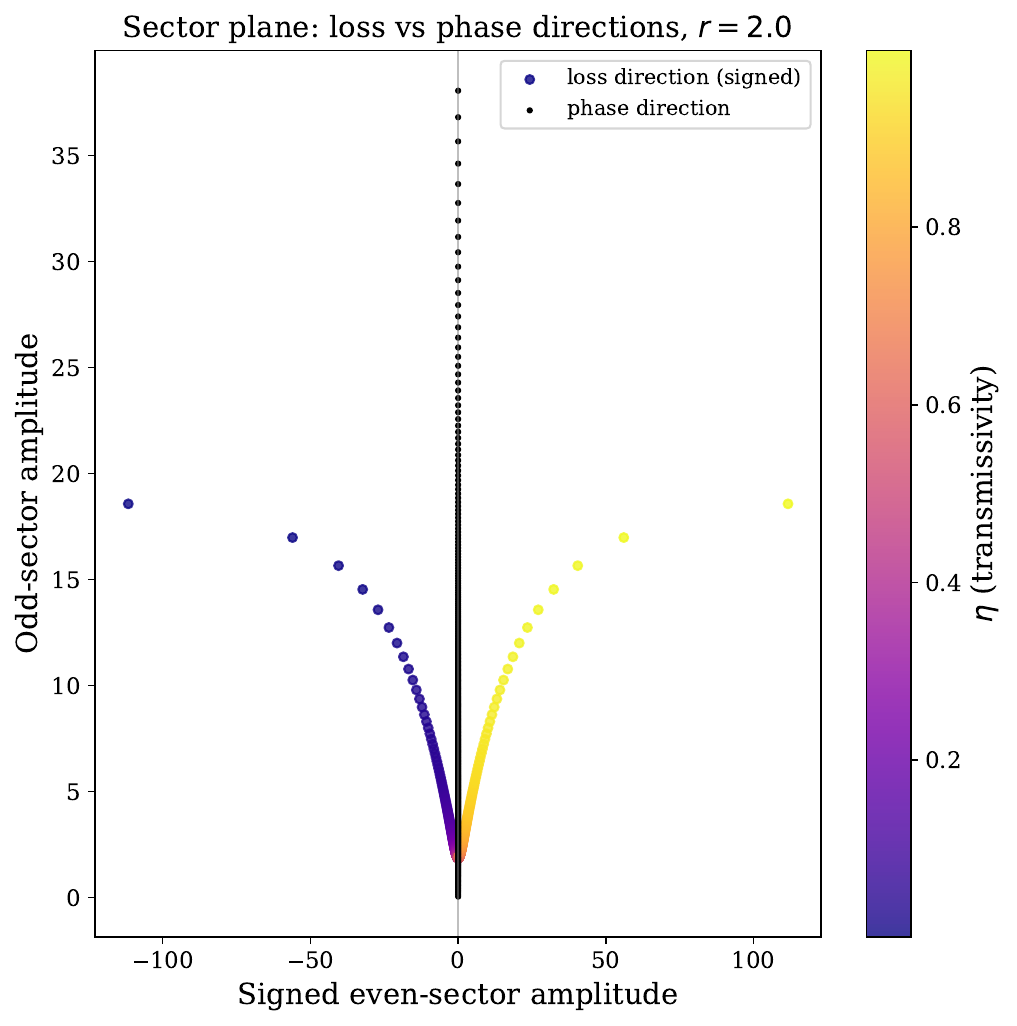}
\caption{\label{fig:epsart} Plot of sector vectors for even and odd in the case of phase estimation.}
\end{centering}
\end{figure}

Even before splitting into even and odd sectors, the diagonal QFIM tells
us that, locally, $\eta$ and $\theta$ can be estimated jointly without
fundamental incompatibility. The even-odd split adds a finer geometric and physical
decomposition of \emph{how} each parameter is encoded in the Gaussian
state.

\paragraph*{(i) Loss parameter $\eta$-purity plus shape:}
The even contribution $(QFIM_e)_{\eta\eta}$ measures
how sensitively the \emph{purity} of the mode responds to changes in $\eta$, whereas the odd contribution $(QFIM_o)_{\eta\eta}$ measures how sensitively the \emph{shape} of the noise ellipse (relative
squeezing between $q$ and $p$) responds.

Near the pure boundaries $\eta\to 0$ and $\eta\to 1$, the symplectic eigenvalue approaches $k=\tfrac12$, and the state becomes pure. Along such
purity-breaking directions $(QFIM_e)_{\eta\eta}\to\infty$, whereas $(QFIM_o)_{\eta\eta}$
remains finite. In this regime, the metrological power for estimating $\eta$ is overwhelmingly purity-based (even sector).

At intermediate losses (e.g.\ $\eta\approx \tfrac12$), the state is strongly mixed and its purity changes slowly with $\eta$, so the even part becomes less important and the odd part (the sensitivity of the squeezing anisotropy to $\eta$) can dominate.

\paragraph*{(ii) Phase parameter $\theta$- purely odd, shape-only:}

The $\theta$ parameter does not change the symplectic eigenvalue $k$ or the purity; it only rotates the noise ellipse in phase space, creating off-diagonal correlations between $q$ and $p$. Accordingly, $({QFIM}_e)_{\theta\theta} = 0$ while
$({QFIM}_o)_{\theta\theta} > 0$.

More explicitly,
\begin{equation}
  ({QFIM}_o)_{\theta\theta}
  = \frac{4\eta^2 \sinh^2(2r)}{4\Lambda_1\Lambda_2 + 1}.
\end{equation}
This is the familiar scaling of phase sensitivity with squeezing: for fixed $\eta$, the QFI in $\theta$ grows initially as $\sinh^2(2r)$ and is penalised by the mixedness factor $4\Lambda_1\Lambda_2+1$ in the denominator.

Thus, the phase parameter is a \emph{purely odd-sector} parameter: all its metrological power derives from shape/correlation changes, with no contribution from purity. This is consistent with the fact that phase
shifts are unitary and preserve the Gaussian symplectic spectrum.

Overall, the even--odd splitting elucidates the joint phase--loss estimation problem in two complementary ways: it (i) identifies $\theta$ as a purely shape/correlation parameter and $\eta$ as a mixed purity+shape parameter, and (ii) explains the orthogonality of the parameters and the structure of the QFIM in terms of the geometry of parity-graded tangent directions on the Gaussian state manifold.

\section{\label{sec:level6}
Conclusion and Outlook \protect}

In this work, we have introduced a geometric and operational decomposition of the QFI for centered multimode Gaussian states into two additive, orthogonal contributions, which we call the even and odd sectors. The construction is based on a Cartan decomposition of the symplectic Lie algebra into generators that are even or odd with respect to the symplectic form, and it is implemented directly at the level of tangent vectors to the Gaussian state manifold. Any infinitesimal change of a covariance matrix can thus be written as the sum of an \textit{even component}, associated with deformations of the symplectic spectrum, and an \textit{odd component}, associated with correlation-generating Gaussian dynamics at fixed spectrum. Because the Bures metric and the corresponding QFI respect this splitting, the total QFI assumes the form $QFI_e+QFI_o$ with each term non-negative and endowed with a distinct geometric meaning and metrological interpretation.

On the manifold of pure Gaussian states, this structure simplifies in a striking way. There, all symplectic eigenvalues are fixed at the minimum allowed by the uncertainty principle, so that even directions are frozen out and the even sector of the Bures metric vanishes identically. The QFI is then entirely odd and directly connected to the natural metric on the Siegel upper half-space that parametrizes pure Gaussian states via their complex adjacency (graph) matrices. In this regime, we derived an explicit graphical expression for the QFI, which shows that \textit{metrological sensitivity} is fully determined by how the parameter drives the underlying Gaussian graph along the \textit{Siegel manifold}. This provides a geometric foundation for continuous-variable metrology with pure Gaussian probes and clarifies which Gaussian resources are actually metrologically relevant.

For mixed Gaussian states, the geometry naturally acquires a fiber-bundle structure. The base manifold consists of admissible symplectic eigenvalues (the spectrum), while the fibers encode the choice of symplectic frame and thus the pattern of correlations. The even sector of the Bures metric depends only on motion along the base, that is, on changes in the \textit{symplectic eigenvalues}, and we obtained closed expressions for this contribution in the \textit{Williamson eigenbasis}. The odd sector is complementary: it is insensitive to purely spectral variations and quantifies precisely those deformations that change the symplectic frame at fixed spectrum. In this way, the odd QFI measures the metrological usefulness of correlation-building Gaussian operations—such as single- and two-mode squeezing and nontrivial mode mixing, while orthogonal symplectic transformations play a controlled and purely odd role. This picture yields a clean separation between spectral and correlation-based resources, which is particularly useful for interpreting and comparing realistic sensing protocols.

We further derived a state-dependent lower bound on the even QFI in terms of the rate of change of the global purity. This result shows that the even sector imposes a \textit{speed limit} on how fast a parameter can change the mixedness of a Gaussian state: given a fixed amount of even QFI, the purity cannot vary arbitrarily quickly along the parameter manifold. Conversely, monitoring how purity changes with the parameter provides an experimentally accessible way to obtain a lower bound on the even contribution to the QFI. The divergence of this bound when approaching the pure manifold reflects the increasing \textit{stiffness} of the metric along even directions and explains why purely spectral sensitivity is suppressed in the pure-state limit.

To illustrate the operational content of the even–odd decomposition, we analyzed several representative metrological scenarios. For unitary encodings such as \textit{ beam-splitter} and \textit{two-mode-squeezing} interactions acting on thermal inputs, the even sector isolates the sensitivity that arises from population and temperature differences, while the odd sector quantifies genuinely correlation-based enhancements. In the context of thermal-contrast interferometry, the even QFI captures how sensitivity depends on temperature imbalance and vanishes when the thermal inputs are identical, whereas the odd QFI becomes dominant when active squeezing generates inter-mode correlations. For Gaussian channels such as loss and phase-insensitive amplification, the decomposition separates spectral changes (e.g., attenuation or amplification of symplectic eigenvalues) from the effect of pre- and post-processing Gaussian unitaries. This allows one to identify regimes where improved performance must ultimately come from correlation-generating dynamics rather than from population engineering alone.

We also extended the analysis to multi-parameter quantum metrology by constructing the full QFI matrix and studying its block structure. Parameters that affect only the spectrum (for example, loss or temperature) and parameters that act purely through correlation-generating unitaries (such as phase shifts in an appropriate frame) naturally define even and odd directions, respectively. In joint phase–loss estimation, this leads to QFI matrices that are close to block diagonal in suitable coordinates, revealing when the two parameters can be estimated compatibly and when trade-offs are unavoidable. From a practical viewpoint, the even–odd decomposition thus provides a diagnostic tool for protocol design: it indicates when resources should be invested in cooling, noise and channel engineering (to enhance the even contribution), and when genuine advantages require more sophisticated multimode Gaussian operations that enrich the odd sector.

The present framework opens several promising directions for future research. A natural extension is the extension to non-Gaussian states and operations, where one may ask whether analogous symmetry-induced decompositions of the Bures metric persist beyond the Gaussian setting. From a practical perspective, the even--odd split suggests concrete design principles for continuous-variable sensors: protocols that aim to suppress noise sensitivity should minimize even-sector contributions, while correlation-enhancing strategies should be optimized to maximize the odd sector under experimental constraints.

Finally, we believe this geometric viewpoint to be useful beyond quantum metrology, for instance in benchmarking Gaussian channels, analyzing dissipative phase transitions in bosonic systems, and studying resource theories where purity and correlations play distinct operational roles. By making explicit the geometric and algebraic structure underlying Gaussian quantum Fisher information, our results provide a versatile framework for understanding and exploiting quantum-enhanced sensitivity in continuous-variable platforms.

\color{black}
\begin{acknowledgments}
K.C. and U.L.A. gratefully acknowledge support from
the Danish National Research Foundation, Center for Macroscopic Quantum States (bigQ,\ DNRF0142), 
EU project CLUSTEC (grant agreement no. 101080173), and EU ERC project ClusterQ (grant agreement no. 101055224, ERC-2021-ADG). TP acknowledges research funding from the QVLS-Q1 consortium, supported by the Volkswagen Foundation and the Ministry for Science and Culture of Lower Saxony.
P.C. acknowledges the support from the International Postdoctoral Fellowship from the Ben May
Center for Theory and Computation. V. S.  is supported by a KIAS Individual Grant No. PG096801   at Korea Institute for Advanced Study. 
\end{acknowledgments}

\appendix
\section{Proofs for Sec. III}\label{app1}
Most of the theorems and claims that are made in the main text are concentrated inside section \ref{sec:level4}. The physical interpretation of the claims and the intuition behind the proofs are already described in detail in the main text. Here, we provide the main steps behind the proofs for the interested readers:\\
We start with the proof of theorem \ref{thanos1} which relates QFI and Graphical calculus together. The main idea of the proof is to do a direct evaluation of the Bures metric in terms of graphical parameters:
\begin{proof}
    We proceed with the direct evaluation of $8ds^2_{pure}=\tr[(V^{-1}dV)^2]$ with $V=S_ZS_Z^T=\begin{pmatrix} Y^{-1} & Y^{-1}X\\
XY^{-1}& Y+XY^{-1}X\end{pmatrix}$. Then $V^{-1}=\begin{pmatrix} Y+XY^{-1}X & -XY^{-1}\\
-Y^{-1}X& Y^{-1}\end{pmatrix}$. Now we first get the $dV$ part:
\[
\begin{aligned}
d(Y^{-1}) &= -\,Y^{-1}(dY)\,Y^{-1},\\
d(Y^{-1}X) &= -\,Y^{-1}(dY)\,Y^{-1}X + Y^{-1}(dX),\\
d(XY^{-1}) &= (dX)\,Y^{-1} - X\,Y^{-1}(dY)\,Y^{-1},\\
d\!\bigl(Y+XY^{-1}X\bigr)
 &= dY + (dX)\,Y^{-1}X \\
 &- X\,Y^{-1}(dY)\,Y^{-1}X + X\,Y^{-1}(dX).
\end{aligned}
\]
Then $V^{-1}dV=\begin{pmatrix} A & B\\
C& D\end{pmatrix}$ where 
\[
\begin{aligned}
A &= -dY\,Y^{-1}-XY^{-1}(dX)Y^{-1}, \\[4pt]
B &= dX - dY\,Y^{-1}X - XY^{-1}dY - XY^{-1}(dX)Y^{-1}X, \\[4pt]
C &= Y^{-1}(dX)Y^{-1}, \\[4pt]
D &= -A^T.
\end{aligned}
\]
Now $\tr[(V^{-1}dV)^2]/2=\tr[A^2+BC]=\tr[(Y^{-1}dX)^2]+\tr[(Y^{-1}dY)^2]=\tr[Y^{-1}dZY^{-1}dZ^*]=ds^2_{Siegel}$. In simplifying such equations we have heavily used cyclicity of trace and invariance of trace under transposition. We have also used the fact that $X,Y,dX,dY$ are all symmetric. 
\end{proof}
Next we discuss the proof behind splitting of the Bures metric in the representation of $(Z,\Gamma)$ as discussed in the main text. The main idea behind this is leveraging the Cartan decomposition of Lie-algebra of symplectic group along with the parity property of commutation or anti-commutation with respect to the symplectic form $\Omega$ . Proof of theorem \ref{x3} proceeds as:
\begin{proof}
    Given a state $(Z,\Gamma)$, we have the covariance matrix as $V=S_Z\Gamma S^T_Z$, and we evaluate the Bures metric in a unitarily rotated frame $S_Z$ frame, where
    \begin{equation}
        ds^2=2\tr[S^{-1}_ZdV S_Z^{-T}(4\mathcal{L}_{\Lambda}+\mathcal L_\Omega)^{-1}(S^{-1}_ZdV S^{-T}_Z)]
    \end{equation}
    From standard algebra it is easy to see that $d\bar V:= S^{-1}_ZdV S^{-T}_Z=w\Gamma+\Gamma w^T+d\Gamma=[a,\Gamma]+\{s,\Gamma\}+d\Gamma$ where we identify $dv_o=\{s,\Gamma\}; dv_e=d\Gamma+[a,\Gamma]$. Using the fact that $a^T=-a,  s^T=s$ and both are elements of $\mathfrak{sp}_{2n}$ we get
    \be
    \begin{split}
        \Omega dv_e&=dv_e\Omega\\
        \Omega dv_o&=-dv_o\Omega
    \end{split}
    \ee
This commutation property is what we leverage. It is important to note that this property holds only because of the way we split our state information as $(Z,\Gamma)$ such that $[\Gamma,\Omega]=0$. We define $P_\Omega(.)=\Omega(.)\Omega^T$ and then call any operator $H$ even if $P_\Omega(H)=H$ and odd in $P_\Omega(H)=-H$. Then it is easy to see that the operator $\mathcal M_\Gamma\circ P_\Omega=P_\Omega\circ \mathcal M_\Gamma$ and using properties of pseudo-inverse one shows that $\mathcal M_\Gamma^{-1}$ also commutes with $P_\Omega$. Overall this means that cross terms of the form $\tr[dv_e\mathcal M_\Gamma^{-1}(dv_o)]=0$:
\be
\begin{split}
   \tr[dv_e\mathcal M_\Gamma^{-1}(dv_o)]&=-\tr[dv_e\mathcal M_\Gamma^{-1}\circ P_\Omega(dv_o)]\\
   &=-\tr[dv_eP_\Omega\circ\mathcal M_\Gamma^{-1}(dv_o)]\\
   &= -\tr[dv_e\mathcal M_\Gamma^{-1}(dv_o)]
\end{split}
\ee
and this simplification leads to the expression as required. 
\end{proof}
Next we discuss the proof of proposition \ref{prop:frame-invariance} which states that the splitting is invariant and hence, canonically defined over set of frames connected via orthogonal symplectic matrices. The main idea behind the proof is to leverage the property that any element $O\in O_{sp}$ commutes with the symplectic form $\Omega$. 
\begin{proof}
By definition,
\[
  \dot\Sigma_1 = S_1^{-1}\dot V S_1^{-T},\qquad
  \dot\Sigma_2 = S_2^{-1}\dot V S_2^{-T}.
\]
Using $S_2 = S_1 O$ and the fact that $O(t)$ is not differentiated in
$\dot V$, we have
\[
  \dot\Sigma_2
  = O(t)^T\,\dot\Sigma_1\,O(t)
  =: P_{O(t)}(\dot\Sigma_1),
\]
where $P_O(X):= O^T X O$ is an orthogonal linear map on the space of
symmetric matrices. The even/odd projectors are
\[
  P_\pm(X) := \tfrac12\bigl(X \pm P_\Omega(X)\bigr),\qquad
  P_\Omega(X) := \Omega X \Omega^T.
\]
Since $O(t)\in O_{sp}$ commutes with $\Omega$, we have
$P_\Omega\circ P_{O} = P_{O}\circ P_\Omega$, and hence
$P_\pm\circ P_{O} = P_{O}\circ P_\pm$. Thus
\[
  P_\pm(\dot\Sigma_2)
  = P_\pm\bigl(P_O(\dot\Sigma_1)\bigr)
  = P_O\bigl(P_\pm(\dot\Sigma_1)\bigr).
\]

The Bures superoperator for $K$ transforms under the conjugation
$P_O$ as
\[
  M_{V_2} = P_O \circ M_{V_1} \circ P_O^{-1},
\]
so, by spectral calculus for self-adjoint operators (or using transformation properties of pseudo-inverse under singular value decomposition) we get,
\[
  M_{V_2}^{-1} = P_O \circ M_{V_1}^{-1} \circ P_O^{-1}.
\]
The even QFI in frame~2 is therefore
\begin{align}
  QFI_e^{(2)}
  &= 2\,\mathrm{Tr}\bigl[
        P_+(\dot\Sigma_2)\,
        M_{V_2}^{-1}\bigl(P_+(\dot\Sigma_2)\bigr)
      \bigr] \\
  &= 2\,\mathrm{Tr}\bigl[
        P_O P_+(\dot\Sigma_1)\,
        P_O M_{V_1}^{-1} P_O^{-1}\,
        P_O P_+(\dot\Sigma_1)
      \bigr] \\
  &= 2\,\mathrm{Tr}\bigl[
        P_+(\dot\Sigma_1)\,
        M_{V_1}^{-1}\bigl(P_+(\dot\Sigma_1)\bigr)
      \bigr] \\
  &= QFI_e^{(1)},
\end{align}
where we used the orthogonality of $P_O$ with respect to the Hilbert--Schmidt
inner product, i.e.\ $\mathrm{Tr}[P_O(X) P_O(Y)] = \mathrm{Tr}[XY]$.
The same argument applies to the odd projector $P_-$, yielding
$\mathrm{QFI}_o^{(2)} = \mathrm{QFI}_o^{(1)}$.
\end{proof}
Now we discuss how to compute the QFI splitting in Willamson's frame which is the main content of theorem \ref{will decom}:
\begin{proof}
    For the proof, we use the fact that $\dot\Sigma=X_++X_-$ has a decomposition as in equation (\ref{thanos2}), where we name the even term as $X_+$ and the odd one as $X_-$. We also use $V=diag(K,K)$ with $K=diag(k_1,\ldots,k_n)$. Now we compute the action of $\mathcal M_V$ on $\dot \Sigma$. We only show the even part because the odd part is similar. 
    \begin{itemize}
        \item $\mathcal{L}_\Omega(X_+)=-X_+$
        \item $\mathcal{L}_V(X_+)=\begin{pmatrix} KMK & -KNK\\
KNK& KMK\end{pmatrix}$
\item Overall, this means 
\begin{equation}
\begin{split}
    \mathcal{M}_V(X_+)&=\begin{pmatrix} \alpha^-\circ M & -\alpha^-\circ N\\
\alpha^-\circ N& \alpha^-\circ M\end{pmatrix}\\
\mathcal{M}_V(X_-)&=\begin{pmatrix} \alpha^+\circ A & \alpha^+\circ B\\
\alpha^+\circ B& -\alpha^+\circ A\end{pmatrix}
\end{split}
\end{equation}
\item Now for the pseudo-inverse, it is sufficient to notice that $\mathcal{M}_V(\mathcal{B}_{ij})=(4k_ik_j-1)\mathcal{B}_{ij}$ where $\mathcal B_{ij}$ can be a basis element for the space of symmetric matrices or anti-symmetric matrices within each block. This means that on such spaces the matrix $\mathcal M_V$ is diagonal and from this we can get its inverse. More precisely, $\mathcal B^{\pm}_{ij}=\ket{i}\bra{j}\pm \ket{j}\bra{i}$ where we can use the plus for spanning the symmetric matrices and the minus for spanning the anti-symmetric ones. 
\end{itemize}
Here, $M\circ N$ is the element-wise Hadamard product between matrices $M,N$. This gives the pseudo-inverse terms as:
\begin{equation}
    \begin{split}
    \mathcal{M}^{-1}_V(X_+)&=\begin{pmatrix} \alpha^-\oslash M & -\alpha^-\oslash N\\
\alpha^-\oslash N& \alpha^-\oslash M\end{pmatrix}\\
\mathcal{M}^{-1}_V(X_-)&=\begin{pmatrix} \alpha^+\oslash A & \alpha^+\oslash B\\
\alpha^+\oslash B& -\alpha^+\oslash A\end{pmatrix}
\end{split}
\end{equation}

where $(M\oslash N)_{ij}=N_{ij}/M_{ij}$ is elementwise Hadamard division. Because of the pseudo-inverse, we are safe for cases where the denominator becomes 0. Now computing $2\tr[X_+\mathcal M^{-1}_V(X_+)]$ and $2\tr[X_-\mathcal M^{-1}_V(X_-)]$ gives the above equation (11).
\end{proof}
Below we sketch the proof of theorem \ref{thanos3} which gives a state-dependent lower bound on
the even QFI in terms of the rate of change of the global
purity. This result shows that the even sector imposes
a speed limit on how fast a parameter can change the
mixedness of a Gaussian state: given a fixed amount
of even QFI, the purity cannot vary arbitrarily quickly
along the parameter manifold.:
\begin{proof}
Here, the main idea is that from equation \ref{will decom} we have $QFI_e(t)$ as a sum of weighted Frobenius norms of matrix $M$ and $N$. Now we can extract the squares of the diagonal terms of $M$, which gives us a lower bound as: 
\begin{align}
        \frac{1}{4}QFI_e(t)&\geq \sum_i\frac{M_{ii}^2}{a_{ii}^-}\\
        & \geq\sum_i\frac{\dot k^2_i}{4k^2_i-1}\geq 4\frac{<\vec v,\vec a>^2}{||a||^2}
\end{align}
    
where $v_i=\frac{\dot k_i}{\sqrt{4k^2_i-1}}$ and $a_i=\frac{\sqrt{4k^2_i-1}}{k_i}$ and we used Cauchy-Schwarz inequality. Now observe that $\dot{ln(\mu)}=-<\vec v,\vec a>$ and $||V^{-1}||^2=2\sum_i\frac{1}{k_i^2}$. This ends the proof.
\end{proof}
\bibliography{apssamp}
\end{document}